\documentclass[12pt]{article} 
\pdfoutput=1
\usepackage[authoryear]{natbib}
\usepackage{graphicx,epstopdf}
\usepackage{setspace}
\usepackage[margin=2.5cm]{geometry}
\usepackage{array}
\usepackage{hyperref}

\usepackage{pdflscape}
\usepackage{graphicx}
\usepackage[utf8]{inputenc}
\usepackage[english]{babel}
\usepackage{amsmath, amsthm, amssymb, amsfonts}
\newtheorem{theorem}{Theorem}[section]
\newtheorem{corollary}{Corollary}[theorem]
\newtheorem{lemma}[theorem]{Lemma}
\theoremstyle{definition}

\usepackage{float}
\floatstyle{ruled}
\newfloat{algorithm}{tbp}{loa}
\providecommand{\algorithmname}{Algorithm}
\floatname{algorithm}{\protect\algorithmname}

\theoremstyle{plain}
\newtheorem{assumption}{\protect\assumptionname}
\theoremstyle{definition}
 \newtheorem{example}{\protect\examplename}
\theoremstyle{plain}
\newtheorem{thm}{\protect\theoremname}
\providecommand{\assumptionname}{Assumption}
\providecommand{\examplename}{Example}
\providecommand{\theoremname}{Theorem}

\usepackage{ifpdf}
\ifpdf
  \DeclareGraphicsExtensions{.pdf,.png,.jpg}
\else
  \DeclareGraphicsExtensions{.eps}
\fi

\usepackage{natbib}
\setcitestyle{authoryear,open={(},close={)}}

\usepackage{enumitem}
\usepackage{subfig}
\usepackage[section]{placeins}
\usepackage{framed}

\begin{document}
\begin{singlespacing}
\title{Scenario Sampling for Large Supermodular Games}

\maketitle
\medskip{}

\begin{center}
{\large{}Bryan S. Graham \textcircled{r} and Andrin Pelican}\footnote{{\footnotesize{}\underline{Graham}}{\footnotesize{}: Department of
Economics, University of California - Berkeley, 530 Evans Hall \#3380,
Berkeley, CA 94720-3880 and National Bureau of Economic Research,
}{\footnotesize{}\underline{e-mail:}}{\footnotesize{} \href{http://bgraham@econ.berkeley.edu}{bgraham@econ.berkeley.edu},
}{\footnotesize{}\underline{web:}}{\footnotesize{} \url{http://bryangraham.github.io/econometrics/}.
}{\footnotesize{}\underline{Pelican:}}\:{\footnotesize{}\underline{e-mail:}}{\footnotesize{} \href{http://pelicanandrin@gmail.com}{pelicanandrin@gmail.com}.
Author names were randomly ordered using the American Economic Association's Author Randomization Tool \url{https://www.aeaweb.org/journals/policies/random-author-order/generator}; confirmation code \textrm{iDdPvyEyJea0}. See \cite{Ray_Robson_AER2018} for discussion and background. Financial support from NSF (SES \#1357499) as well as the Spanish State Research Agency under the María de Maeztu Unit of Excellence Programme (Project No: CEX2020-001104-M) is gratefully acknowledged. Feedback from participants at the CEMMAP/Toulouse conference on games (June 2021), the Oxford panel data conference (June 2022), the CEMMAP/Vanderbilt conference on games (May 2023), the IMS/SNAB conference (June 2023), as well as at seminars at UC Berkeley, UCL, Warwick, LSE, Oxford, and CEMFI is gratefully acknowledged. This revision reflects specific feedback and questions from Bo Honor\'e, \'Aureo de Paula, Adam Rosen and Steve Berry. All the usual disclaimers apply.}}
\par\end{center}

\begin{center}
\medskip{}
\textsc{\large{}Initial Draft: March 2021}{\large\par}
\par\end{center}

\begin{center}
\textsc{\large{}This Draft: July 2023}{\large\par}
\par\end{center}
\begin{abstract}

This paper introduces a simulation algorithm for evaluating the
log-likelihood function of a large supermodular binary-action game. 
Covered examples include (certain types of) peer effect, technology adoption, strategic network formation, and multi-market entry games.
More generally, the algorithm facilitates simulated maximum likelihood (SML)
estimation of games with large numbers of players, $T$, and/or
many binary actions per player, $M$ (e.g., games with tens of thousands
of strategic actions, $TM=O(10^4)$). In such cases the likelihood of the
observed pure strategy combination is typically (i) very small and (ii)
a $TM$-fold integral who region of integration has a complicated geometry. Direct
numerical integration, as well as accept-reject Monte Carlo integration,
are computationally impractical in such settings. In contrast, we 
introduce a novel importance sampling algorithm
which allows for accurate likelihood simulation with modest numbers
of simulation draws. \smallskip{}

\end{abstract}
\end{singlespacing}

\thispagestyle{empty} 

\pagebreak{}

\setcounter{page}{1}

\section{Introduction}

Payoff interdependence, where the utility one agent gets from taking
a particular action varies with the actions chosen by others, characterizes
many models of economic behavior. The returns to adopting a networked
technology increase in the fraction of other agents also adopting \citep[e.g., ][]{Goolsbee_Klenow_JLE2002,Ackerberg_Gowrisankaran_RAND2006}.
A teen's decision to smoke may be influenced by whether her friends also smoke \citep[e.g., ][]{Gaviria_Raphael_RESTAT2001,Krauth_JBES2007,Card_Giuliano_RESTAT2013}. The returns to forming an R\&D partnership (or a supply chain relationship) between two firms might vary with the presence or absence of such partnerships (relationships) across other firms. Firms' decisions regarding market entry typically depend on the entry decisions of competing firms \citep[e.g., ][]{Bresnahan_Reiss_JOE91, Jia_EM08,Ciliberto_Tamer_EM2009}. In these settings, and many others, it is convenient to view the set of actions taken by agents as an equilibrium of a game.
\\
\\
The econometric analysis of games poses a number of, now well understood, challenges (see, for example, the surveys by \cite{dePaula_ARE13},  \cite{Aradillas_Lopez_ARE2020} and \cite{Molinari_HBE2020}). 
One set of issues stems from model \emph{incompleteness}: econometric models of strategic interaction often admit multiple Nash Equilibria (NE) for a given set of (unobserved) preference shocks, $\mathbf{U}$, and (unknown) payoff parameters, $\theta$. Another set of issues involves computational tractability. Even if the econometrician ``completes'' the model by specifying
which equilibrium is played when (or, more generally, writes down an equilibrium selection model), parameter estimation in complete information games with more than a handful of players and actions is often infeasible because evaluating the likelihood function involves a high-dimensional integral with a complex region of integration. 
\\
\\
In this paper we introduce an approach to payoff parameter estimation appropriate for large supermodular complete information games with binary actions. By ``large" we mean games with hundreds of players, each of whom might take hundreds of of actions (i.e., games with tens of thousands of strategic variables). This class of games includes many examples of interest to empirical researchers in economics and other fields.
\\
\\
Consider the peer effects in smoking example introduced above (with $t=1,\ldots,T$ players). A researcher might postulate that the payoff from smoking is increasing in the fraction of one's peers who smoke. The payoff may also vary with observed agent attributes, $X_t$, and an unobserved normally (or logistically) distributed random utility ``shock", $U_t$. Under complete information, any pure strategy Nash Equilibrium (NE) in this setting will coincide with a fixed point of a system of $T$ non-linear simultaneous equations. A consequence of this simultaneous determination of smoking decisions is that an agent's unobserved taste for smoking -- the Gaussian random utility shock in her payoff function -- will co-vary with the fraction of her peers that smoke. The naive probit regression fit of own smoking behavior onto own attributes and the fraction of peers smoking does not consistently recover agent preferences \citep[e.g., ][]{Heckman_EM1978}.
\\
\\
If the researcher is willing to make an equilibrium selection assumption, we will assume that the NE equilibrium with the fewest number of smokers is the one that prevails, then the likelihood is well-defined. Evaluating this likelihood, because it requires computing a $T$-fold integral, is difficult. For modest $T$, say 20 or 30 agents, \cite{Krauth_JOE2006} and \citet{Soetevent_Kooreman_JAE07} have proposed simulation algorithms for computing this integral. These algorithms make use of additional structure in this example beyond supermodularity. Our approach handles this example as well as others, including those where, say, $T=500$, and each agent takes $M>1$ actions. In such settings likelihood evaluation requires computing a $TM$-fold integral with a very complex region of integration.
\\
\\
For researchers unfamiliar with the econometric complexities of game estimation, our method is relatively straightforward to adopt. For example, a peer effects researcher can simply introduce the fraction of one's peers who smoke as a ``covariate" in a probit-like specification. In typical examples computation requires no more than a few minutes on a good desktop computer. Unlike the naive probit fit described above, our method -- by appropriately handling the game-theoretic aspects of the model -- delivers consistent estimates of payoff function parameters under easy to communicate assumptions.
\\
\\
For economists with experience in game estimation, our approach makes large game analysis feasible. For example, we show how our methods can be used to estimate payoff functions in some models of (directed) network formation. In this example each of $T$ agents decides whether to direct a link (or not) to each of the $T-1$ other agents in her network. This results in a game with $T(T-1)=O(T^2)$ strategic variables. In an empirical illustration we study the Nykatoke network dataset collected by \cite{deWeerdt_IAP04}. This network includes $T=116$ households. Following \cite{deWeerdt_IAP04} we fit a model which includes household specific sender and receive fixed effects, observed household attributes, and a taste for supported links (see \cite{Jackson_et_al_AER12}). We are aware of no extant methods, beyond those introduced below, that would allow a researcher to empirically study a game with over $13,000$ strategic decision variables with a payoff function indexed by over two hundred parameters.\footnote{We defer on a discussion of whether fitting such a high dimensional model to the Nyakatoke network is sensible to later in the paper.} 
\\
\\
Computational limits have profoundly shaped the nature of empirical work involving strategic interaction. While methods that fully embrace payoff interdependence and strategic interaction commonly feature in empirical industrial organization, where many settings of interest involve just a few agents \citep[e.g.,][]{Ciliberto_Tamer_EM2009},
applications involving many agents and/or actions are scarce. The failure to embrace the ``strategic nature'' of discrete interactions in, for example, empirical peer effects analyses arguably undermines the credibility of work in this area. Indeed, this was one theme of \citet{Manski_ReStud93}. We emphasize peer effects and network formation examples below.
\\
\\
An important accomplishment of econometricians studying games has been the development of methods of inference that do not require making \emph{a priori} assumptions on equilibrium selection (see \cite{Molinari_HBE2020} for a comprehensive survey). We depart from this norm and \emph{do} maintain an equilibrium selection assumption in what follows (see also \cite{Bajari_et_al_EM10}). Our results are also restricted to supermodular games (loosely games where the actions of others do not discourage -- or weakly encourage -- a player to take an action). While they are restrictive relative to some treatments in prior work, these assumptions allows us to focus on our main contribution: likelihood evaluation for very large games. We discuss how to relax our equilibrium selection assumption at the close of the paper.
\\
\\
In a large game, evaluating the likelihood involves computing a high-dimensional integral over a complex region of integration. Our approach involves transforming this integral into an expectation and using importance sampling to estimate it. The use of importance sampling in the simulated maximum likelihood (SML) context goes back at least to \cite{McFadden_EM89} in the econometrics literature. Our particular problem requires computing the sum of a (typically very large) set of ``rectanglar'' probabilities (i.e., hypercube volumes). Similar to the GHK algorithm of \cite{Geweke_EM1989}, \cite{Hajivassiliou_Ruud_HBE1994}, \cite{Keane_EM94} and \cite{Hajivassiliou_et_al_JOE1996}, we generate random vectors within a target rectangle by taking draws from a sequence of univariate truncated distributions. Unlike the GHK simulator, the location of the rectangles whose volume we require is not known ex ante and the points of truncation are determined sequentially via game-theoretic arguments. We also compute individual rectangle probabilities ``analytically" as opposed to using a weighted frequency estimator. We elaborate on these and other differences below.
\\
\\
To be clear, our innovation is the introduction of a particular importance sampler: \emph{scenario sampling}. The likelihood of a game outcome, say $\mathbf{Y}=\mathbf{y}=\left(y_1, \dots y_T \right)'$, conditional on a set of agent attributes, $\mathbf{X}=\left(X_1, \dots X_T \right)'$, and parameter value, $\theta$, coincides with the probability that the set of random utility shocks, $\mathbf{U}=\left(U_1, \dots U_T \right)'$, lie in a region where $\mathbf{Y}=\mathbf{y}$ is the selected NE. In a binary action game this region will be a collection of high-dimensional hyper-cubes or ``scenarios". The total volume of these cubes equals the likelihood. We estimate this volume by sampling scenarios and aggregating them in a particular way.
\\
\\
In order to sample a scenario in the target set, we need to construct a vector of random utility draws $\mathbf{U}$ such that $\mathbf{Y}=\mathbf{y}$ is the selected NE with probability one. We accomplish this task by drawing the agent-by-action random utility shocks sequentially. The support of a given draw may depend on the realizations of prior draws. Such an approach allows us to ensure that, in the end, $\mathbf{U}$ will be such that $\mathbf{Y}=\mathbf{y}$ is the selected NE. Our sequential sampler uses both the structure of the NE conditions as well as supermodularity.
\\
\\
The methods introduced below make estimation of a binary peer effects games with hundreds of peer groups, each consisting dozens of agents, relatively
routine. Similarly we outline a set of tools that would allow a researcher to easily fit an economically interesting structural model of strategic network formation to a graph consisting of hundreds of agents. We are aware of no comparable estimation methods for these settings.\footnote{The closest competitor would be the simulated method-of-moments (SMM) approach outlined by \cite{Uetake_Watanabe_AE13} and used by \cite{Jia_EM08}, \cite{Nishida_MS2015} and \cite{Miyauchi_JOE16}. This approach involves comparing moments calculated from simulated game outcomes to those observed in the dataset in hand. Depending on the application of interest, this approach can be of comparable computational cost to ours. In other settings it can be more costly. For example, in the context of network formation models, it is well-known that computing induced subgraph frequencies -- the natural moments to use for SMM estimation in this context -- is computationally very costly \citep[e.g.,][]{Bhattacharya_Bickel_AS15, Graham_HBE2020}. If the target parameter $\theta$ has more than a few components, then our method -- as it allows for gradient-based optimization -- also has possible speed advantages. For some models we are also able to avoid repeated NE computation, something that is required by the SMM-method (see Appendix \ref{app: recycling}).  Finally, our approach, being likelihood-based, also offers efficiency advantages and facilitates Bayesian inference (for interested researchers).} 
\\
\\
The next section introduces a simple, and likely familiar, coordination game. We use this game to introduce the main assumptions and features of our methods in an easy to understand way. 
\\
\\
Section \ref{sec: peer effects model} extends our basic analysis to binary peer effects games. The analysis of such games is commonplace, but extant empirical work generally ignores, or side-steps, game theoretic difficulties (but see \cite{Krauth_JOE2006,Krauth_JBES2007} and \cite{Soetevent_Kooreman_JAE07} for important exceptions). Section \ref{sec: discrete supermodular games} extends our results to a general class of discrete supermodular games. Certain games of technology adoption and network formation, as we explain, are members of this class.
\\
\\
In Section \ref{sec: monte carlo experiments} we explore the numerical properties of our methods via a series of Monte Carlo experiments. Finally, in Section \ref{sec: nykatoke}, we use our methods to estimate the payoff parameters in a game of network formation using the Nyakatoke network data collected by \cite{deWeerdt_IAP04}. Our empirical model of network formation is inspired by the ``support" model of favor exchange introduced by \cite{Jackson_et_al_AER12}.
\\
\\
In what follows random variables are denoted by capital Roman letters, specific realizations by lower case Roman letters and their support by blackboard bold Roman letters. That is $Z$, $z$ and $\mathbb{Z}$ respectively denote a generic random draw of, a specific value of, and the support of, $Z$. Random vectors and matrices are generally written in boldface (e.g., $\mathbf{X}$). We use Greek letters for parameters and a ``0" subscript to denote their population values. We sometimes omit this subscript when doing so causes no confusion.
\\
\\
\section{\label{sec: coordination game}A simple coordination game}
We begin with a simple game of coordination between a pair of friends $t=1,2$. Each friend/agent takes a binary action $Y_{t}\in \left\{ 0,1\right\}$. In \citet{Card_Giuliano_RESTAT2013} $Y_{t}$ corresponds to adolescent behaviors such as sexual intercourse, smoking, marijuana use or chronic
truancy. To keep things concrete (and light) we will consider two
friends, Ademaro ($t=1$) and Brunhilde ($t=2$), who are deciding whether
to attend ($Y_{t}=1$) an electronic dance music (EDM) concert or
not ($Y_{t}=0$).
\\
\\
The payoff function equals
\begin{equation}
\upsilon\left(y_{t},y_{-t};X_{t},U_{t},\theta\right)=y_{t}\left(X_{t}'\beta+\delta y_{-t}-U_{t}\right)\label{eq: utility_2player_example}
\end{equation}
for $t=1,2$. The payoff from not attending the concert is normalized to zero for both agents. Observable differences in the taste for EDM are captured by variation in the linear index $X_t'\beta$. Unobserved heterogeneity is captured by the random utility shifter $U_t$. Note the negative sign in front of $U_t$; hence it measures an agent's unobserved \emph{distaste} for EDM. The notation $y_{-t}$ denotes the actions of players other than $t$.
\\
\\
As we assume that $\delta\geq0$, the marginal utility
associated with attending the concert, $Y_{t}=1$, is greater when
your friend also attends, $Y_{-t}=1$. This makes the game supermodular. 
Ademaro's utility from attending, $Y_{1}=1$, when Brunhilde does
not, $Y_{2}=0$, is
\[
v\left(1,0;X_{1},U_{1},\theta\right)=X_{1}'\beta-U_{1}.
\]
Whereas if Brunhilde also attends his utility increases by $\delta$
to
\[
v\left(1,1;X_{1},U_{1},\theta\right)=X_{1}'\beta+\delta-U_{1}.
\]
Similarly the utility Brunhilde receives from attending the concert
depends on whether Ademaro does not attend
\begin{align*}
v\left(1,0;X_{2},U_{2},\theta\right)=X_{2}'\beta-U_{2}
\end{align*}
or does attend
\[
v\left(1,1;X_{2},U_{2},\theta\right)=X_{2}'\beta+\delta-U_{2}.
\]

\subsubsection*{Scenarios}
The systematic utility of taking the action, $X_{t}'\beta+\delta y_{-t}$, when evaluated at all possible combinations of peer play, $y_{-t} \in \left\{0,1\right\}$, defines a partition of the support of $U_t$ into what we call \emph{buckets} \citep[cf., ][]{Pelican_Graham_NBER2020}; also see Figure
\ref{fig: understanding-scenarios}. The bucket partition for the support of $U_{t}$, for $t=1,2$, is
\begin{align*}
\mathbb{R}= & \left(-\infty,X_{t}'\beta\right]\cup\left(X_{t}'\beta,X_{t}'\beta+\delta\right]\cup\left(X_{t}'\beta+\delta,\infty\right).
\end{align*}
The number of buckets will coincide with the number of possible peer play case distinctions, $L$, plus $1$. Specifically, the bucket partition (and its cardinality) can be found mechanically by evaluating the utility function for all possible values of $y_{-t}$.
\\
\\
In our EDM example $L=2$, such that there are three buckets. If $U_1$ falls into the first bucket, then Ademaro's taste shock is sufficiently low that it is strictly dominant for him to attend the concert. That is he will go irrespective of what Brunhilde chooses to do. If $U_1$ instead falls into the second, or middle, bucket, then Ademaro is on the fence. His $U_1$ realization is low enough that it will be optimal for him to attend the concert if Brunhilde does as well, but he will not go to the concert without Brunhilde. Finally if $U_1$ falls into the third bucket, then it is a strictly dominant strategy for Ademaro to \emph{not} go to the concert: his distaste for EDM is so high that even the presence of Brunhilde at his side is not enough to entice him to go. The interpretation of Brunhilde's bucket partition is analogous.
\\
\\
A pair of buckets, one from Ademaro and one from Brunhilde, defines what we call a \emph{scenario}. Let $\mathbf{b}_{jk}$
denote the \emph{scenario} where $U_{1}$ falls into bucket $j=1,2,3$
and $U_{2}$ falls into bucket $k=1,2,3$. A scenario is a region of the support of $\mathbf{U}=\left(U_1,U_2\right)'$ where the fundamental strategic considerations of agents are constant. In other words, a scenario defines a range of realizations for the utility/cost
shocks $\mathbf{U}=\left(U_{1},U_{2}\right)'$ where, within them,
the fundamental nature of strategic play does not depend on the precise
values of $U_{1}$ or $U_{2}$.
\\
\\
In Figure \ref{fig: understanding-scenarios} scenario
$\mathbf{b}_{22} = \left(X_{1}'\beta, X_{1}'\beta+\delta \right]   \times \left(X_{2}'\beta, X_{2}'\beta+\delta \right]$ corresponds to a pair of random utility draws $\mathbf{U}=\left(U_1,U_2\right)'$ where both Ademaro and Brunhilde are ``on the fence" about going to the concert. That is where it is a NE for them to both go or to both not go. This conclusion does not depend on the precise values of $U_{1}$ and $U_{2}$, only that they
are somewhere within scenario $\mathbf{b}_{22}.$
\\
\\
In scenario $\mathbf{b}_{22}$ there are two NE, $\mathbf{Y}=\left(0,0\right)'$ and $\mathbf{Y}=\left(1,1\right)'$; the model is incomplete. In what follows \emph{we complete the model} by assuming that when $\mathbf{U} \in \mathbf{b}_{22}$ Ademaro and Brunhilde do not go the concert (that is that $\mathbf{Y}=\left(0,0\right)'$ is the selected NE). This is a special case of our assumption, stated formally below, that the equilibrium with the ``least" amount of action is the one that prevails if $\mathbf{U}$ falls in a scenario that supports multiple NE.
\\
\\
The set of all scenarios, denoted by $\mathbb{B}$, partitions the support of $\mathbf{U}=\left(U_1,U_2\right)'$ into a set of nine rectangles:
\begin{align*}
\mathbb{R}^2= 
                          & \mathbf{b}_{11}  \cup \mathbf{b}_{12}  \cup \mathbf{b}_{13}  \cup \mathbf{b}_{21}  \cup \mathbf{b}_{22}  \cup \mathbf{b}_{23}  \cup \mathbf{b}_{31}  \cup \mathbf{b}_{32}  \cup \mathbf{b}_{33}.\\
\end{align*} See Figure \ref{fig: understanding-scenarios} for an illustration \citep[see also][]{Bresnahan_Reiss_JOE91,Tamer_ReStud03,dePaula_ARE13}). For the time being we use double subscripts to index different scenarios. As a second example of a scenario, when $U_{1}$ lies in its first bucket, and likewise for $U_{2}$, then Ademaro and Brunhilde are in scenario $\mathbf{b}_{11} = \left(\infty, X_{1}'\beta\right]   \times \left(\infty, X_{2}'\beta\right]$
(corresponding to the lower-left-hand rectangle in Figure \ref{fig: understanding-scenarios}).
In this scenario Ademaro's (Brunhilde's) utility/cost shock is so low
that he (she) will attend the EDM concert irrespective of
whether Brunhilde (Ademaro) does. In this scenario both players' strictly
dominant strategy is to attend the concert; $\mathbf{Y}=\left(1,1\right)'$ is the NE.
\\
\\
\subsubsection*{Likelihood}
With an equilibrium selection assumption in hand, the probability
of any game outcome $\mathbf{Y}=\mathbf{y}=\left(y_{1},y_{2}\right)'$
simply corresponds to the probability that $\mathbf{U}=\left(U_{1},U_{2}\right)$
falls into one of the scenarios in which $\mathbf{Y}=\mathbf{y}$
is the (selected) NE. We denote the subset of scenarios where $\mathbf{Y}=\mathbf{y}$ is the NE by $\mathbb{B}_{\mathbf{y}}$. The set of \emph{all} scenarios is denoted by $\mathbb{B}$.
\begin{figure}
\caption{\label{fig: understanding-scenarios}Scenarios in an EDM concert attendance
game}

\begin{centering}
\includegraphics[scale=0.55]{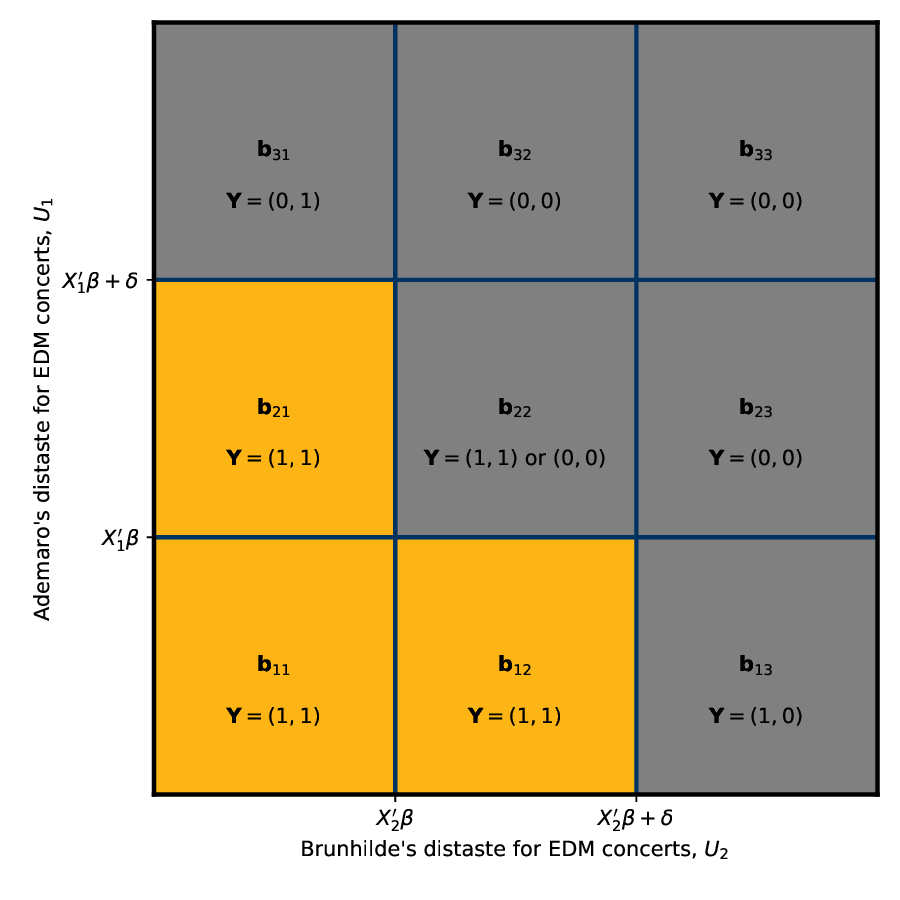}
\includegraphics[scale=0.55]{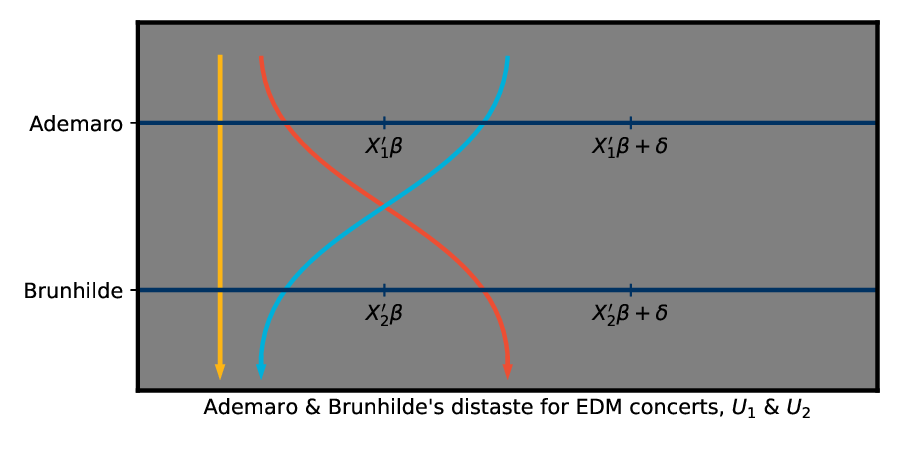}
\par\end{centering}
\underline{Notes:} \newline (i) The figure in the top panel shows the partition of $\mathbb{R}^2$, the
support of $\mathbf{U}=(U_1,U_2)'$, into nine scenarios. The three darker shaded scenarios are those where the NE is for Ademaro and Brunhilde to both attend the concert. The middle scenario, labelled $\mathbf{b}_{22}$, corresponds to the case where they are both ``on the fence" about going. In this scenario both  $\mathbf{Y}=(0,0)'$ and $\mathbf{Y}=(1,1)'$ are NE. We resolve this incompleteness by assuming that $\mathbf{Y}=(0,0)'$ is the selected equilibrium. This corresponds to the NE with the least amount of action. The likelihood of the event ``Ademaro and Brunhilde both go to the concert" is given by the probability mass attached to the three darker shaded regions \newline
(ii) The bottom figure shows all sequences of $U_1$ and $U_2$ that ``land" in one of the three scenarios where the selected NE is $\mathbf{Y}=(1,1)'$. The left-most line in the figure depicts a sequence where both $U_1$ and $U_2$ are low such that is strictly dominant for them both to go. The second case corresponds to when Ademaro gets a low shock and Brunhilde a medium one. In this case, although Brunhilde is on the fence, she nevertheless goes to the concert because Ademaro goes (it is strictly dominant for him to do so). The last sequence corresponds to Ademaro getting a medium shock and Brunhilde a low one.
\end{figure}
\\
\\
For example, the probability of observing $\mathbf{Y}=\left(1,1\right)'$
in a randomly sampled EDM concert attendance game (from some well-defined
population of EDM concert attendance games), corresponds to the ex
ante chance that a pair of random utility shocks falls into one of the three darker shaded regions 
of Figure \ref{fig: understanding-scenarios} (i.e., into scenarios $\mathbf{b}_{11}$, $\mathbf{b}_{12}$ or $\mathbf{b}_{22}$):
\begin{align}\label{eq: coordination_likelihood}
\Pr\left(\left.\mathbf{Y}=\tbinom{1}{1}\right|\mathbf{X};\theta\right)
= & \sum_{\mathbf{b}\in\mathbb{B}_{\mathbf{y}}}\int_{\mathbf{u}\in \mathbf{b}}f_{\mathbf{U}}\left(\mathbf{u}\right)\mathrm{d}\mathbf{u} \\
=  & \int_{\mathbf{u}\in\mathbf{b}_{11}}f_{\mathbf{U}}\left(\mathbf{u}\right)\mathrm{d}\mathbf{u}+\int_{\mathbf{u}\in\mathbf{b}_{12}}f_{\mathbf{U}}\left(\mathbf{u}\right)\mathrm{d}\mathbf{u}+\int_{\mathbf{u}\in\mathbf{b}_{21}}f_{\mathbf{U}}\left(\mathbf{u}\right)\mathrm{d}\mathbf{u}\nonumber \\
= & F\left( X_{1}'\beta \right)F\left(X_{2}'\beta \right) +
    F\left(X_{1}'\beta \right)\left[F\left(X_{2}'\beta+\delta \right) - 
    F\left(X_{2}'\beta\right)\right]\notag \\
  & +\left[F\left(X_{1}'\beta + \delta \right)-F\left(X_{1}'\beta \right)\right] F\left(X_{2}'\beta\right),\notag
\end{align}
where we assume that $U_{1}$ and $U_{2}$ are iid with known CDF $F\left(\cdot\right)$
and PDF $f\left(\cdot\right)$ such that $f_{\mathbf{U}}\left(\mathbf{u}\right)=f\left(u_{1}\right)f\left(u_{2}\right)$.
\\
\\
The three summands in \eqref{eq: coordination_likelihood} correspond to the probability mass attached to each of the three colored scenarios in Figure \ref{fig: understanding-scenarios}. In this simple two player game, with two-dimensional scenarios, direct likelihood evaluation involves no difficulties. Consequently maximum likelihood estimation (MLE) is both straightforward and entirely standard.
\\
\\
However, consider the direct extension of the game to accommodate three players. In such a game there would be four buckets and $4^3=64$ scenarios (corresponding to cubes in $\mathbb{R}^3$).\footnote{If agents are not exchangeable, for example peers are ``best friends" and ``second best friends", then there would be $5$ buckets and $125$ scenarios.} In general, the number of scenarios for an observed game outcome will grow exponentially with the number of players/strategic decisions.\footnote{In games with additional special structure, the number of scenarios may grow more slowly with $T$.} In this paper we are interested in large games. Those with many players, $T$, each of whom, might take many binary actions, $M$. When $TM$ is in the hundreds or thousands, direct likelihood evaluation, and hence MLE, is not feasible.
\\
\\
Our approach to (approximate) likelihood evaluation in many scenario games involves simulation. The probability that a random draw of $\mathbf{U}=\left(U_{1},U_{2}\right)'$ falls in scenario $\mathbf{b}$ is simply
\begin{equation*}
    \zeta\left(\mathbf{b};\theta\right)\overset{def}{\equiv}\int_{\mathbf{u}\in \mathbf{b}}f_{\mathbf{U}}\left(\mathbf{u}\right)\mathrm{d}\mathbf{u},
    \end{equation*}
where we suppress the role of covariates, $\mathbf{X}$, in the notation. For example, the ex ante probability that Ademaro and Brunhilde find themselves in scenario $\mathbf{b}_{22}$ is
\begin{align*}    \zeta\left(\mathbf{b}_{22};\theta\right) = & \int_{u_{1}=X_{1}'\beta}^{X_{1}'\beta+\delta}\int_{u_{2}=X_{2}'\beta}^{X_{2}'\beta+\delta}f\left(u_{1}\right)f\left(u_{2}\right)\mathrm{d}u_{1}\mathrm{d}u_{2}\\
 = & \left[F\left(X_{1}'\beta+\delta\right)-F\left(X_{1}'\beta\right)\right]\left[F\left(X_{2}'\beta+\delta\right)-F\left(X_{2}'\beta\right)\right].
\end{align*}
Observe that $\zeta\left(\mathbf{b};\theta\right)$ is a pmf for scenarios with support $\mathbb{B}$. Let $\mathbf{B}$ denote a random draw from the distribution of scenarios described by this pmf. We can re-write the likelihood of the event $\mathbf{Y}=\tbinom{1}{1}$, that is equation \eqref{eq: coordination_likelihood}, as
\begin{equation*}
    \Pr\left(\left.\mathbf{Y}=\left(1,1\right)'\right|\mathbf{X};\theta\right)
= \sum_{\mathbf{b}\in\mathbb{B}_{\mathbf{y}}}\zeta\left(\mathbf{b};\theta\right) = \Pr\left(\mathbf{B}\in\mathbb{B}_{\mathbf{y}}\right).
\end{equation*}
Conceptually, the probability to the right of the last equality above is straightforward to simulate (practically, as we will see, there are difficulties). To see this note that every random draw, $\mathbf{U}$, from the population distribution of preference shocks, $F_{\mathbf{U}}$, will fall into one, and only one, scenario. The event $\mathbf{U} \in \mathbf{b}$ occurs with an ex ante probability of $\zeta\left(\mathbf{b};\theta\right)$. It is therefore easy to generate random draws from $\zeta\left(\mathbf{b};\theta\right)$ because the population distribution of $\mathbf{B}$ over 
$\mathbb{B}$ is induced by the one for the random utility shifters $\mathbf{U}$ (which are easy to simulate). Hence we have the equality
\begin{equation*}
    \Pr\left(\mathbf{B}\in\mathbb{B}_{\mathbf{y}}\right)=\Pr\left(\mathbf{y}\text{ is \emph{the} NE at }\mathbf{U}\right),
\end{equation*}
where the probability to the right can be computed by the accept/reject Monte Carlo (“dartboard”) simulation estimate
\begin{equation*}
\hat{\Pr}\left(\left.\mathbf{Y}=\mathbf{y}\right|\mathbf{X};\theta\right)=\frac{1}{S}\sum_{s=1}^{S}\mathbf{1}\left(\mathbf{y}\text{ is \emph{the} NE at }\mathbf{U}^{\left(s\right)}\right)    
\end{equation*}
with $s=1, \dots, S$ indexing independent random draws $\mathbf{U}^{(s)}$ from $F_{\mathbf{U}}$.
\\
\\
Unfortunately, in large games, it is generally the case that the event $\mathbf{y}$ is a NE at $\mathbf{U}$ occurs with very low probability. In large games the number of possible pure strategy combinations, the cardinality of the set $\mathbb{Y}$, is typically enormous (here $\mathbb{Y}$ is the set of all $2^{TM}$ possible game outcomes). The population probability of observing any particular game outcome or NE, $\mathbf{Y}=\mathbf{y}$, is therefore very small. Estimating small probabilities with a computationally manageable number of simulation draws, $S$, by accept/reject frequency methods is well-known to be infeasible (see, for example, \cite{Hajivassiliou_Ruud_HBE1994} or \cite{Au_Beck_PEM2001}).
\\
\\
As in other related contexts our solution to this conundrum involves importance sample. We begin with the observation that evaluating the likelihood of a given game outcome $\mathbf{Y}=\mathbf{y}$ only requires integration over scenarios in the set $\mathbb{B}_{\mathbf{y}}$ (i.e., scenarios where $\mathbf{Y}=\mathbf{y}$ \emph{is} the selected NE). Those scenarios in the complement $\mathbb{B}\setminus\mathbb{B}_{\mathbf{y}}$ do not enter the likelihood calculation. Although enumeration of the set $\mathbb{B}_{\mathbf{y}}$ is infeasible in large games, we show how it is possible to sample randomly from it. 
\\
\\
Let $\lambda_{\mathbf{y}}\left(\mathbf{b};\theta\right)$ be a function which assigns probabilities to the scenarios contained in $\mathbb{B}_{\mathbf{y}}$. We will require that $\lambda_{\mathbf{y}}\left(\mathbf{b};\theta\right)$ be strictly greater than zero for any $\mathbf{b}\in\mathbb{B}_{\mathbf{y}}$ and exactly zero otherwise (i.e., for $\mathbf{b}\in\mathbb{B}\setminus\mathbb{B}_{\mathbf{y}})$. We also require that this function satisfy the adding up condition $\sum_{\mathbf{b}\in\mathbb{B}_{\mathbf{y}}}\lambda_{\mathbf{y}}\left(\mathbf{b};\theta\right)=1$. The function $\lambda_{\mathbf{y}}\left(\mathbf{b};\theta\right)$ is a pmf for those scenarios where $\mathbf{Y}=\mathbf{y}$ is the NE. Let $\tilde{\mathbf{B}}$ be a random scenario draw from the distribution with pmf $\lambda_{\mathbf{y}}\left(\mathbf{b};\theta\right)$. Later we show how to construct such a draw, but for now assume an appropriate method is in hand. Note that, by construction, $\Pr\left(\mathbf{\tilde{B}}\in\mathbb{B}_{\mathbf{y}}\right)=1$ and $\Pr\left(\mathbf{\tilde{B}}\in\mathbb{B}\setminus\mathbb{B}_{\mathbf{y}}\right)=0$.
\\
\\
Let $\theta^{\left(0\right)}$ be some (fixed) value for the parameter; we have that
\begin{align*}
\Pr\left(\left.\mathbf{Y}=\mathbf{y}\right|\mathbf{X};\theta\right)	= & \sum_{\mathbf{b}\in\mathbb{B}_{\mathbf{y}}}\zeta\left(\mathbf{b};\theta\right) \\
	= & \sum_{\mathbf{b}\in\mathbb{B}_{\mathbf{y}}}\frac{\zeta\left(\mathbf{b};\theta\right)}{\lambda_{\mathbf{y}}\left(\mathbf{b};\theta^{\left(0\right)}\right)}\lambda_{\mathbf{y}}\left(\mathbf{b};\theta^{\left(0\right)}\right) \\
	= & \mathbb{E}_{\tilde{\mathbf{B}}}\left[\frac{\zeta\left(\tilde{\mathbf{B}};\mathbf{X},\theta\right)}{\lambda_{\mathbf{y}}\left(\tilde{\mathbf{B}};\theta^{\left(0\right)}\right)}\right],
\end{align*}
where $\tilde{\mathbf{B}}$ denotes a random draw from $\lambda_{\mathbf{y}}\left(\mathbf{b};\theta^{\left(0\right)}\right).$ We have written $\Pr\left(\left.\mathbf{Y}=\mathbf{y}\right|\mathbf{X};\theta\right)$ as an expectation over scenarios in $\mathbb{B}_{\mathbf{y}}$ (versus a summation over the much larger set $\mathbb{B})$. An importance sampling Monte Carlo estimate of this expectation is:
\begin{equation}\label{eq: importance_likelihood_simulator}    \hat{\Pr}\left(\left.\mathbf{Y}=\mathbf{y}\right|\mathbf{X};\theta\right)=\frac{1}{S}\sum_{s=1}^{S}\frac{\zeta\left(\tilde{\mathbf{B}}^{\left(s\right)};\theta\right)}{\lambda_{\mathbf{y}}\left(\tilde{\mathbf{B}}^{\left(s\right)};\theta^{\left(0\right)}\right)},
\end{equation}
where $\tilde{\mathbf{B}}^{\left(1\right)} \dots \tilde{\mathbf{B}}^{\left(S\right)}$ are independent random draws from $\lambda_{\mathbf{y}}\left(\mathbf{b};\theta^{\left(0\right)}\right).$
\\
\\
This estimate, because the cardinality of $\mathbb{B}_{\mathbf{y}}$ is finite, is consistent as $S\rightarrow\infty$ (see below). More importantly, because all of the summands in \eqref{eq: importance_likelihood_simulator} are non-zero, this estimate has the potential to provide precise estimates of $\Pr\left(\left.\mathbf{Y}=\mathbf{y}\right|\mathbf{X};\theta\right)$ for modest values of $S$, particularly if the importance sampling weights, $\lambda_{\mathbf{y}}\left(\mathbf{b};\theta\right)$, are close to uniform.
\\
\\
Operationalizing \eqref{eq: importance_likelihood_simulator} requires a method for sampling scenarios in $\mathbb{B}_{\mathbf{y}}$; such a method is the primary contribution of this paper. We first describe our approach in the context of the simple two player coordination game introduced in this section and generalize it to larger games in the sequel. Our approach is to draw $U_1$ and $U_2$ \emph{sequentially} such that $\mathbf{U}=\left(U_1, U_2 \right)'$ is in a scenario in $\mathbb{B}_{\mathbf{y}}$ with probability one (i.e., $\mathbf{U} \in \tilde{\mathbf{B}},\,\, \tilde{\mathbf{B}} \in \mathbb{B}_{\mathbf{y}}$ w.p.1).
\\
\\
Consider simulating the probability of the event that both Ademaro and Brunhilde go to the concert (i.e., that $\mathbf{Y}=\mathbf{y}=\tbinom{1}{1}$). For the purposes of illustration, we will first draw Brunhilde's preference shock, $U_2$, followed by Ademaro's, $U_1$. In order to construct a draw of $\mathbf{U}$ that falls into one of the three darker shaded regions of Figure \ref{fig: understanding-scenarios}, it is necessary, albeit not sufficient, that $U_2 \leq X_{2}'\beta+\delta$. In the first step of our procedure we therefore draw $U_2$ from $F$ truncated at $X_{2}'\beta+\delta$.
\\
\\
Next we draw Ademaro's preference shock, $U_1$. Our approach to doing so depends on the realized value of Brunhilde's shock, $U_2$. If Brunhilde's shock falls between $X_{2}'\beta$ and $X_{2}'\beta+\delta$, then she is ``on the fence". It will only be a NE for her to go if Ademaro does as well. This means that Ademaro's shock \emph{must} fall below $X_{1}'\beta$, such that it is strictly dominant for him to go. So in this case we draw $U_1$ from $F$ truncated at $X_{1}'\beta$. This generates 
$\left(U_1, U_2 \right)' \in \mathbf{b_{12}}$.
\\
\\
Alternatively, consider the case where Brunhilde's step one shock was sufficiently low such that it is strictly dominant for her to go to the concert; that is $U_2 \leq X_{2}'\beta$. In this case we are free to draw Ademaro's shock from the larger interval $U_1 \leq X_{1}'\beta+\delta$. If $X_{1}'\beta \leq U_1 \leq X_{1}'\beta+\delta$, then Ademaro is ``on the fence", but since Brunhilde is going no matter what, they both will go. We have $\left(U_1, U_2 \right)' \in \mathbf{b_{21}}$. If, instead, $U_1 \leq X_{1}'\beta$, then we have $\left(U_1, U_2 \right)' \in \mathbf{b_{11}}$ and they both go in that case as well.
\\
\\
Note that finding the appropriate region of support for Ademaro's shock involves the following ``counterfactual" thought experiment. In step 1 we simulate Brunhilde's shock under the presumption that, in the end, Ademaro will also go to the concert (we are computing the probability that they both go). In step 2, after drawing Brunhilde's shock, $U_2$, we ask ourselves what her play would be if Ademaro's taste shock instead was high enough such that it would be strictly dominant for him not to go to the concert. Since the game is supermodular, Ademaro not going (weakly) discourages Brunhilde from attending. We compute the NE associated with this ``counterfactual" scenario and use it to find the appropriate upper threshold for Ademaro's shock. If Brunhilde is sensitive to Ademaro's choice (i.e., she is ``on the fence" or $U_2$ is in her middle bucket), then we force Ademaro's shock to be lower. If Brunhilde is insensitive to Ademaro's choice (i.e., it is strictly dominant for her to go or $U_2$ is in her first bucket), then we can allow Ademaro's shock to range higher.
\\
\\
The above procedure ensures that $\left(U_1, U_2 \right)'$ falls in a scenario in $\mathbb{B}_{\mathbf{y}}$ with probability one. It also reaches every scenario in this set with positive probability. Formalizing the calculations above yields the (importance) sampling probabilities
\begin{align} \label{eq: lambdas_coordination_ex1}
\lambda_{\mathbf{y}}\left(\mathbf{b}_{11};\theta^{\left(0\right)}\right)	= & \frac{F\left(X_{1}'\beta^{\left(0\right)}\right)}{F\left(X_{1}'\beta^{\left(0\right)}+\delta^{\left(0\right)}\right)}\left[\frac{F\left(X_{2}'\beta^{\left(0\right)}\right)}{F\left(X_{2}'\beta^{\left(0\right)}+\delta^{\left(0\right)}\right)}\right]\\
\lambda_{\mathbf{y}}\left(\mathbf{b}_{12};\theta^{\left(0\right)}\right)	= &  \frac{F\left(X_{2}'\beta^{\left(0\right)}+\delta^{\left(0\right)}\right)-F\left(X_{2}'\beta^{\left(0\right)}\right)}{F\left(X_{2}'\beta^{\left(0\right)}+\delta^{\left(0\right)}\right)}\\
\lambda_{\mathbf{y}}\left(\mathbf{b}_{21};\theta^{\left(0\right)}\right)	= & \left[\frac{F\left(X_{1}'\beta^{\left(0\right)}+\delta^{\left(0\right)}\right)-F\left(X_{1}'\beta^{\left(0\right)}\right)}{F\left(X_{1}'\beta^{\left(0\right)}+\delta^{\left(0\right)}\right)}\right]\frac{F\left(X_{2}'\beta^{\left(0\right)}\right)}{F\left(X_{2}'\beta^{\left(0\right)}+\delta^{\left(0\right)}\right)}.
\end{align}
It is a simple exercise to verify that:
\begin{equation*}
    \lambda_{\mathbf{y}}\left(\mathbf{b}_{11};\theta^{\left(0\right)}\right) + \lambda_{\mathbf{y}}\left(\mathbf{b}_{12};\theta^{\left(0\right)}\right)
    + \lambda_{\mathbf{y}}\left(\mathbf{b}_{21};\theta^{\left(0\right)}\right)
    = 1,
\end{equation*}
and hence that we have defined a proper probability distribution that places positive weight on all scenarios in $\mathbb{B}_{\mathbf{y}}$.
In contrast the unconditional population frequencies of these three scenarios are (when $\theta_0=\theta$):
\begin{align*}
    \zeta\left(\mathbf{b}_{11};\theta\right) = & F\left(X_{1}'\beta\right)F\left(X_{2}'\beta\right)\\
    \zeta\left(\mathbf{b}_{12};\theta\right) = & F\left(X_{1}'\beta\right)\left[F\left(X_{2}'\beta+\delta\right)-F\left(X_{2}'\beta\right)\right]\\
    \zeta\left(\mathbf{b}_{21};\theta\right) = & \left[F\left(X_{1}'\beta+\delta\right)-F\left(X_{1}'\beta\right)\right]F\left(X_{2}'\beta\right).    
\end{align*}
It is helpful for understanding what follows to observe that the scenario sampling probabilities are inverse probability weighted versions of their population counterparts:
\begin{align*}
\lambda_{\mathbf{y}}\left(\mathbf{b}_{11};\theta^{\left(0\right)}\right) = & \frac{1}{F\left(X_{1}'\beta^{\left(0\right)}+\delta^{\left(0\right)}\right)F\left(X_{2}'\beta^{\left(0\right)}+\delta^{\left(0\right)}\right)}\zeta\left(\mathbf{b}_{11};\theta^{\left(0\right)}\right)\\
\lambda_{\mathbf{y}}\left(\mathbf{b}_{12};\theta^{\left(0\right)}\right) = & \frac{1}{F\left(X_{1}'\beta^{\left(0\right)}\right)F\left(X_{2}'\beta^{\left(0\right)}+\delta^{\left(0\right)}\right)}\zeta\left(\mathbf{b}_{12};\theta^{\left(0\right)}\right)\\
\lambda_{\mathbf{y}}\left(\mathbf{b}_{21};\theta^{\left(0\right)}\right) = & \frac{1}{F\left(X_{1}'\beta^{\left(0\right)}+\delta^{\left(0\right)}\right)F\left(X_{2}'\beta^{\left(0\right)}\right)}\zeta\left(\mathbf{b}_{22};\theta^{\left(0\right)}\right).
\end{align*}
Observing that $\hat{\lambda}_{\mathbf{y}}\left(b;\theta^{\left(0\right)}\right) = \frac{1}{S}\sum_{s=1}^{S}\mathbf{1}\left(\tilde{\mathbf{B}}^{\left(s\right)}=\mathbf{b}\right)$ consistently estimates $\lambda_{\mathbf{y}}\left(\mathbf{b};\theta^{\left(0\right)}\right)$ allows us to express our importance sampling likelihood estimate as
\begin{align*}
\hat{\Pr}\left(\left.\mathbf{Y}=\mathbf{y}\right|\mathbf{X};\theta\right) = & 	\frac{1}{S}\sum_{s=1}^{S}\frac{\zeta\left(\mathbf{\tilde{\mathbf{B}}};\theta\right)}{\lambda_{\mathbf{y}}\left(\mathbf{\tilde{\mathbf{B}}};\theta^{\left(0\right)}\right)}\\
 = &	\sum_{\mathbf{b}\in\mathbb{B}_{\mathbf{y}}}\frac{1}{S}\sum_{s=1}^{S}\mathbf{1}\left(\tilde{\mathbf{B}}^{\left(s\right)}=\mathbf{b}\right)\frac{\zeta\left(\mathbf{b};\theta\right)}{\lambda_{\mathbf{y}}\left(\mathbf{b};\theta^{\left(0\right)}\right)}\\
 = &	\sum_{\mathbf{b}\in\mathbb{B}_{\mathbf{y}}}\left\{ \frac{\hat{\lambda}_{\mathbf{y}}\left(\mathbf{b};\theta^{\left(0\right)}\right)}{\lambda_{\mathbf{y}}\left(\mathbf{b};\theta^{\left(0\right)}\right)}\right\} \zeta\left(\mathbf{b};\theta\right)\\
 = & \Pr\left(\left.\mathbf{Y}=\mathbf{y}\right|\mathbf{X};\theta\right)+\sum_{\mathbf{b}\in\mathbb{B}_{\mathbf{y}}}\left\{ \frac{\hat{\lambda}_{\mathbf{y}}\left(\mathbf{b};\theta^{\left(0\right)}\right)-\lambda_{\mathbf{y}}\left(\mathbf{b};\theta^{\left(0\right)}\right)}{\lambda_{\mathbf{y}}\left(\mathbf{b};\theta^{\left(0\right)}\right)}\right\} \zeta\left(\mathbf{b};\theta\right), 
\end{align*}
which indicates that likelihood estimate is unbiased for any fixed $S$ and consistent as $S\rightarrow\infty$ as long as $\lambda_{\mathbf{y}}\left(\mathbf{b};\theta^{\left(0\right)}\right)>0$ for all $\mathbf{b} \in \mathbb{B}_\mathbf{y}$. While $\hat{\Pr}\left(\left.\mathbf{Y}=\mathbf{y}\right|\mathbf{X};\theta\right)$ converges in mean square to $\Pr\left(\left.\mathbf{Y}=\mathbf{y}\right|\mathbf{X};\theta\right)$ at rate $\frac{1}{S}$, accuracy in actual applications will depend on the cardinality of $\mathbb{B}_{\mathbf{y}}$ as well as the features of $\lambda_{\mathbf{y}}\left(\mathbf{b};\theta^{\left(0\right)}\right)$ and $\zeta\left(\mathbf{b};\theta\right)$.

\subsubsection*{Simulated maximum likelihood (SML) estimation}

Let $(\mathbf{X}_{1},\mathbf{Y}_{1}),\ldots,(\mathbf{X}_{N},\mathbf{Y}_{N})$ be a random sample of games. The simulated log likelihood equals:
\begin{equation}\label{eq: simulated_logl}    l_{N}\left(\theta\right)=\sum_{i=1}^{N}\ln\hat{\Pr}\left(\left.\mathbf{Y}_{i}\right|\mathbf{X}_{i};\theta\right),
\end{equation}
with each of the summands in \eqref{eq: simulated_logl} constructed via simulation as described above. When the researcher has access to a large number of independent games (such that the summands in \eqref{eq: simulated_logl} are independent of one another), it is straightforward to characterize the large sample properties of our SML estimator. Under regularity conditions it will be consistent and asymptotically normal as $N\rightarrow\infty,\ S\rightarrow\infty$. See \cite{Hajivassiliou_Ruud_HBE1994}.
\\
\\
We detail additional features of our SML estimator below, but point out now that one advantage (besides the critical one of feasibility!) of importance sampling versus crude frequency simulation is that our criterion function is differentiable. This means that maximization of \eqref{eq: simulated_logl} by gradient-based methods is possible; in turn allowing the dimension of $\theta$ to be large. The Hessian can also be computed numerically, making inference -- when valid -- straightforward.
\\
\\

\section{\label{sec: peer effects model}Peer effect games}
In this section we extend our simple coordination game analysis to $T$-player single-action peer effect and technology adoption games. In these games the payoff from taking action is weakly-increasing in the number of people in an agent's reference, or peer, group who take action \citep[e.g.,][]{Manski_ReStud93}. This is an important class of games. \cite{Manski_ReStud93}, \cite{Brock_Durlauf_HBE01} and \cite{Bramoulle_et_al_JOE09} study the econometrics of these games when actions are continuous and agents' best response functions are linear (see \cite{Bramoulle_et_al_AR2020} for a survey). The binary-action case is less well understood, but see \cite{Krauth_JOE2006} and \cite{Soetevent_Kooreman_JAE07} for important analyses. In Section \ref{sec: discrete supermodular games} we show how to extend our analysis to a class of $T$-player, $M$-binary-action supermodular games.
\\
\\
\subsection{Likelihood}
Consider a simple binary-action peer effects game. There are $T$ agents, each of who decide whether to take an action, $Y_{t}=1$, or not, $Y_{t}=0$. Agents are connected via a set of undirected relationships. Let $\mathbf{G}$ be the $T\times T$ row-normalized adjacency matrix recording this link structure. Let $\mathbf{G}_{t}$ denote the $t^{th}$ row of this matrix such that $\mathbf{G}_{t}\mathbf{y}$ equals the mean action of agent $t$'s peers. The payoff from action for agent $t$ is
\begin{equation} \label{eq: peer_effects_utility}
\upsilon\left(y_{t},\mathbf{y}_{-t};X_{t},U_{t},\theta\right)=y_{t}\left(X_{t}'\beta+\mathbf{G}_{t}\mathbf{y}\delta-U_{t}\right).    
\end{equation}
This set-up approximates many peer effects and technology adoption games. We assume that $\delta\geq0$, so that peer action weakly encourages own action.
\\
\\
The set of pure strategy NE in this game correspond to the solutions of the system of $T$ nonlinear equations
\begin{equation}\label{eq: tarski_system}
Y_{t}=\mathbf{1}\left(X_{t}'\beta+\mathbf{G}_{t}\mathbf{Y}\delta\geq U_{t}\right),\thinspace t=1,\ldots,T.    
\end{equation}
This system may have multiple solutions. In this paper we will assume that the minimal equilibrium, the one with the fewest agents taking action, is the one that prevails. This equilibrium can be found by substituting the zero vector $\mathbf{Y}=\underline{0}_{T}$ into the right-hand-side of \eqref{eq: tarski_system} and iterating to a fixed point. By Tarski's Fixed Point Theorem this iterative process stops at the minimal equilibrium \citep{Milgrom_Roberts_EM1990}.
\\
\\
This equilibrium selection assumption may not be plausible in some applications, it is perhaps most natural for ``opt in” games. For example in adolescent peer effects and technology adoption games it is common for agents to start in the non-action state and then choose whether to act/adopt. Our results are easily modified to accommodate applications where the equilibrium with the most action is the one chosen (the so-called maximal equilibrium). The maximal equilibrium can be found by fixed point iteration starting at $\mathbf{Y}=\underline{1}_{T}$ (i.e., at the one vector). This equilibrium selection assumption might be appropriate for ``opt out” games.
\\
\\
As in the two-player coordination game described above, we can use the systematic component of agents' utility functions \eqref{eq: peer_effects_utility} to partition the support of $U_t$ into buckets. If player $t$ has $J_t$ friends (corresponding to the number of non-zero elements in the $t^{th}$ row of the row-normalized adjacency matrix, $\mathbf{G})$, then her bucket partition will be
\begin{equation*}
    \left(-\infty,X_{t}'\beta\right]\cup\left(X_{t}'\beta,X_{t}'\beta+\frac{\delta}{J_t}\right]\cup\left(X_{t}'\beta+\frac{\delta}{J_t},X_{t}'\beta+\frac{2\delta}{J_t}\right]\cup\cdots\cup\left(X_{t}'\beta+\frac{J_t\delta}{J_t},\infty\right).
\end{equation*}
The number of buckets partitioning the support of $U_{t}$ equals the cardinality of the set $\left\{ s_{t}\left(\mathbf{y}_{-t}\right):\mathbf{y}_{-t}\in\left\{ 0,1\right\} ^T\right\}$  plus $1$ where $s_{t}\left(\mathbf{y}_{-t}\right)=\mathbf{G}_{t}\mathbf{Y}$. In this example we have $L_{t}=J_{t}+1$ for a total of $L_{t}+1=J_{t}+2$ buckets. If $U_t$ falls in the first bucket, then it is strictly dominant for player $t$ to take action regardless of what her peers do. If it falls in the second bucket, it is optimal to take action if at least one out of her $J_t$ peers do and so on. The buckets define ranges of realized values of $U_t$ where it is optimal to take action given that different numbers of peers also take action (the right-most bucket defines the set of $U_t$ realizations where it optimal to not take action across \emph{all} possible levels of peer action).
\\
\\
Scenarios in this game correspond to $T$-dimensional hyper-cubes in $\mathbb{R}^T$. Consider the scenario where all of the $U_t$ taste shocks fall into their second buckets. In this scenario it is optimal for all players to take action if at least one of their peers do.
\\
\\
Associated with a scenario, $\mathbf{b}$, is a set of $T$ upper and lower bucket boundaries. We will denote the bucket boundaries in scenario $\mathbf{b}$ for agent $t$ by $\bar{b}_{t}$ and $\underline{b}_{t}$.\footnote{Note that $\underline{b}_{t} = -\infty$ if agent $t$'s bucket in scenario $\mathbf{b}$ is the first (left-most) one and $\bar{b}_{t}=\infty$ if it is the last (right-most) one.}  Using this notation we can write the ex ante probability that $\mathbf{U}=(U_1,U_2,\ldots,U_T)'$ falls into scenario $\mathbf{b}$ as
\begin{equation*}
\Pr\left(\mathbf{U}\in\mathbf{b}\right)=\zeta\left(\mathbf{b};\theta\right)=\prod_{t=1}^{T}\left[F\left(\bar{b}_{t}\right)-F\left(\underline{b}_{t}\right)\right].    
\end{equation*}
We have suppressed the dependence of the bucket boundaries on $X_t$ and $\theta$ in the notation.
\\
\\
To construct an estimate of the likelihood of the event $\mathbf{Y}=\mathbf{y}$ we, in addition to the expression for $\zeta\left(\mathbf{b};\theta\right)$ given above, require a method for randomly drawing a scenario, say $\tilde{\mathbf{B}}$, from the set $\mathbb{B}_{\mathbf{y}}$ with an \emph{ex ante} probability, $\lambda_{\mathbf{y}}\left(\mathbf{b};\theta^{\left(0\right)}\right)$, that is computable. With these ingredients, we can estimate the likelihood -- as noted earlier -- by the simulated analog of the expectation:
\begin{equation*}
    \Pr\left(\left.\mathbf{Y}=\mathbf{y}\right|\mathbf{X};\theta\right)	= \mathbb{E}_{\tilde{\mathbf{B}}}\left[\frac{\zeta\left(\tilde{\mathbf{B}};\mathbf{X},\theta\right)}{\lambda_{\mathbf{y}}\left(\tilde{\mathbf{B}};\theta^{\left(0\right)}\right)}\right].
\end{equation*}
See \eqref{eq: importance_likelihood_simulator} above. We randomly draw the scenario $\tilde{\mathbf{B}}$ by sequentially drawing the random preference shocks $U_1,U_2,\ldots,U_T$ such that, in the end, $\mathbf{U} \in \tilde{\mathbf{B}}$ and $\tilde{\mathbf{B}} \in \mathbb{B}_{\mathbf{y}}$ with probability one. 
\\
\\
\subsection{Scenario Sampling}

In this section we describe a procedure for constructing a random draw, $\tilde{\mathbf{B}}$, from the set $\mathbb{B}_{\mathbf{y}}$. We use the notation $\tilde{\mathbf{B}}$ to emphasize that this draw is not from the population distribution of scenarios, with support, $\mathbb{B}$ and pmf $\zeta\left(\mathbf{b};\theta\right)$, but instead drawn from a distribution with support on the smaller set $\mathbb{B}_{\mathbf{y}}$. Indeed, the cardinality of $\mathbb{B}_{\mathbf{y}}$ will typically be much smaller than that of $\mathbb{B}$. 
\\
\\
Recall that the population distribution of scenarios is induced by the population distribution of the random utility shocks $\mathbf{U}=\left(U_1,U_2, \ldots U_T \right)'$ (as well as the structure of preferences). We similarly sample scenarios from $\mathbb{B}_{\mathbf{y}}$ by constructing draws of the random utility shocks, $\mathbf{U}$. The innovation is to do this in a way such that $\mathbf{U} \in \tilde{\mathbf{B}}$ and $\tilde{\mathbf{B}} \in \mathbb{B}_{\mathbf{y}}$ with probability one. 
\\
\\
Key to our approach is the drawing of the taste shocks \emph{sequentially} instead of independently. Specifically we allow the region from which $U_{t}$ is drawn from to depend on the realizations of earlier draws, $U_s$ in the sequence (where $s<t$). By carefully taking into account the constraints imposed by the target pure strategy combination $\mathbf{y}$ and the definition of a NE, we can ensure that the final sequence of draws $\mathbf{U}$ fulfills $\mathbf{U} \in \tilde{\mathbf{B}}$ and $\tilde{\mathbf{B}} \in \mathbb{B}_{\mathbf{y}}$ with probability one.
\\
\\
The goal is to construct a draw of $\mathbf{U}=\left(U_1,U_2, \ldots U_T \right)'$ such that, at the current parameter value $\theta$, the selected NE is $\mathbf{y}$. Associated with such a draw is a scenario $\tilde{\mathbf{B}} \in \mathbb{B}_{\mathbf{y}}$ which can be used to simulate the likelihood function using equation \eqref{eq: importance_likelihood_simulator} above.
\\
\\
Let $\mathbf{y}$ denote the NE for which we wish to compute the likelihood $\Pr\left(\left.\mathbf{Y}=\mathbf{y}\right|\mathbf{X};\theta\right)$.
Without loss of generality we will assume that the first $T-s$ agents \emph{do not take} the action (i.e., that $y_t=0$ for $t=1,\dots,T-s$) while the last $s$ agents \emph{do take} the action (i.e., that $y_t=1$ for $t=T-s+1,\dots,T$).
\\
\\
We will first draw random utility shocks for the $T-s$ agents who \emph{do not} take action in NE $\mathbf{y}$. For these agents it must be the case that their utility shock $U_t$ is high enough such that it optimal for them not to take action given the play of others $\mathbf{y}_{-t}$. From \eqref{eq: peer_effects_utility} we have that it must be the case that, for $t=1,\ldots,T-s$,
\begin{equation}
X_{t}'\beta+\mathbf{G}_{t}\mathbf{y}\delta < U_{t}  
\end{equation}
or, equivalently, that $U_{t}\in\left(X_{t}'\beta+\mathbf{G}_{t}\mathbf{y}\delta,\infty\right)$ for $t=1,\ldots,T-s$. Shocks of this magnitude are high enough to ensure that these agents' do not wish to take action, which would be a deviation from their behavior in NE $\mathbf{y}$. See \eqref{eq: tarski_system} above.
\\
\\
We next consider agents the $s$ who \emph{do} take action in the target equilibrium, $\mathbf{y}$. Finding the appropriate range restrictions for $U_t$ for the $y_t=1$ agents is more challenging. Since $\mathbf{G}_{t}\mathbf{y}\delta \geq 0$ (by supermodularity), it must be the case that, if $U_t \in\left(-\infty,x_{t}'\beta \right]$ for $t=T-s+ 1,\ldots,T$, then it will be optimal for these agents to choose $y_t=1$ (as desired). Indeed, draws this low ensure that it will be strictly dominant for these agents to take the action. However it is also possible that higher draws of $U_t$ are sufficiently low enough to ensure that these agents will still choose $y_t=1$.
\\
\\
An example helps clarify the issues involved. Consider a setting where $\mathbf{G}$ corresponds to a complete graph (i.e., everyone is connected to everyone else in the group as in the canonical ``linear-in-means" model studied by \cite{Manski_ReStud93}, \cite{Bramoulle_et_al_JOE09} and others). Our goal continues to be the construction of a $\mathbf{U}=\left(U_1,U_2, \ldots U_T \right)'$ sequence where $\mathbf{Y}=\mathbf{y}$ is the NE. We again begin by drawing the $T-s$ preference shocks for the $y_t=0$ agents from the interval $U_{t}\in\left(X_{t}'\beta+\frac{s\delta}{T-1},\infty\right)$. This range restriction is sufficient to ensure that these $T-s$ agents will not want to deviate from $\mathbf{y}$.
\\
\\
Next we need to draw utility shocks for the $s$ agents who we \emph{do} want to take action. Assume, for the purposes of exposition, that the first $s-1$ of these draws, $U_{T-s+1}, \ldots , U_{T-1}$, are so low that it is strictly dominant for these players to choose $y_t=1$. In this case, for the $s^{th}$ action-taking agent, we only require a random utility shock lower than $X_{t}'\beta+\frac{(s-1)\delta}{T-1}$. For this last shock we can draw $U_{T}\in\left(-\infty,X_{t}'\beta+\frac{(s-1)\delta}{T-1}\right)$.
\\
\\
We now consider a slight variation of this example: we assume that the only first $s-2$ draws for the action-taking agents, $U_{T-s+1}, \ldots , U_{T-2}$, are low enough for it to be strictly dominant for them to choose $y_t=1$. Next consider the appropriate range constraint for the $(s-1)^{th}$ action-taking agent's random utility draw. 
\\
\\
In taking this draw we know that the previous $s-2$ agents will take the action for sure, we will also proceed as if the draw for the $s^{th}$ agent will be low enough for her to want to choose $y_t=1$ as well. In this case any draw of $U_{t}\in\left(-\infty,X_{t}'\beta+\frac{(s-1)\delta}{T-1}\right)$ will suffice for the $(s-1)^{th}$ action-taking agent to choose $y_t=1$ as desired. Say our draw of $U_t$ for agent $t=T-1$ (the $(s-1)^{th}$ agent choosing $y_t=1$) is relatively high, specifically in the interval $\left(X_{t}'\beta+\frac{(s-2)\delta}{T-1},X_{t}'\beta+\frac{(s-1)\delta}{T-1}\right).$
\\
\\
In this situation the appropriate way to draw $U_t$ for agent $T$, the $s^{th}$ and final action-taking agent, is delicate. Say we draw $U_t \in \left(X_{t}'\beta+\frac{(s-2)\delta}{T-1},X_{t}'\beta+\frac{(s-1)\delta}{T-1}\right).$ In this case both agents $T-1$ and $T$ would be willing to choose $y_t=1$ as long as the other did, but they won't choose $y_t=1$ if the other chooses $y_t=0$. Since we assume that the NE with the least amount of action prevails in the presence of multiplicity, such a draw would not result in a scenario in our target set of $\mathbb{B}_{\mathbf{y}}$.
\\
\\
To stay on track we require that $U_T$ lies below $X_{T}'\beta+\frac{(s-2)\delta}{T-1}$. In this case agent $T$ will take the action (because $s-2$ agents will choose $y_t=1$ for sure). Since agent $T$ takes the action, it will also be optimal for agent $T-1$, who had a somewhat higher random utility draw, to do so as well. This agent is on the fence, but since $T$ chooses $y_t=1$ she does so as well. We thus have $\mathbf{U} \in \tilde{\mathbf{B}}$ and $\tilde{\mathbf{B}} \in \mathbb{B}_{\mathbf{y}}$ as needed.
\\
\\
In general there will be an agent specific threshold, $h_t$, between $X_{t}'\beta$ and $X_{t}'\beta+\mathbf{G}_{t}\mathbf{y}\delta$ that will be sufficient to keep our algorithm on track. The theshold $h_t$ is such that if $U_{t} \leq h_t $,  then it is possible to construct subsequent draws $U_{t+1}, \ldots , U_T$, such that, in the end, $\mathbf{U} \in \tilde{\mathbf{B}}$ and $\tilde{\mathbf{B}} \in \mathbb{B}_{\mathbf{y}}$, as
desired. If however $U_{t}>h_t$, then this will not be possible. Algorithm \ref{alg: threshold-finder-peer} shows how to find this threshold. Given these thresholds it is relatively
straightforward to sample $\tilde{\mathbf{B}} \in \mathbb{B}_{\mathbf{y}}$. 
\\
\begin{algorithm}
\caption{\label{alg: scenario-sampler-peer}\textsc{Scenario Sampler} (peer effects case)}
\textbf{Inputs:} $\mathbf{z}=\left(\mathbf{x},\mathbf{y}\right)$, \textbf{$\theta$} (i.e., a target pure strategy
combination and a utility/payoff function)
\begin{enumerate}
\item Initialize $\mathbf{U}=\left(U_{1},\ldots,U_{T} \right)'=
\begin{cases}
-\infty & \text{if } y_t = 1 \\
+\infty & \text{otherwise}
\end{cases}
$

\item For $t=1,\ldots,T$
\begin{enumerate}
\item If $y_{t}=0,$ then sample $U_{t} \in \left(X_{t}'\beta+\mathbf{G}_{t}\mathbf{y}\delta,\infty\right)$
from the conditional density $\frac{f\left(u\right)}{1-F\left(X_{t}'\beta+\mathbf{G}_{t}\mathbf{y}\delta \right)}\overset{def}{\equiv}\omega_{t}  f\left(u\right)$.
\end{enumerate}
\item For $t=1,\ldots,T$
\begin{enumerate}
\item If $y_{t}=1$, then 
\begin{enumerate}
\item determine $h_{t}  $ using \textsc{Threshold}($\mathbf{z}$,\textbf{
$\theta$}, $\mathbf{U}  $, $t$);
\item sample $U_{t}  \in\left(-\infty, h_{t}  \right]$
from the conditional density $\frac{f\left(u\right)}{F\left(h_{t}  \right)}\overset{def}{\equiv}\omega_{t}  f\left(u\right).$
\end{enumerate}
\item Find $\tilde{\mathbf{B}}\in\mathbb{B}_{\mathbf{y}}$ such that $\mathbf{U}  \in\tilde{\mathbf{B}}$.
\end{enumerate}
\end{enumerate}
\textbf{Outputs:} The $T \times1$ weight vector $\underline{\omega}  =\left(\omega_{1}  ,\ldots,\omega_{T} \right)'$,
the vector of taste shocks \textbf{$\mathbf{U}  $} and a (random) scenario $\tilde{\mathbf{B}}\in\mathbb{B}_{\mathbf{y}}$.
\end{algorithm}
\\
Algorithm \ref{alg: scenario-sampler-peer} details how $U_{1} ,\ldots,U_{T}$
are sequentially drawn from a series of truncated distributions. Let $\mathbf{U} $ denote a draw produced by Algorithm
\ref{alg: scenario-sampler-peer}. The (ex ante) population probability of the event $\mathbf{U} \in\mathbf{b}$
(for $\mathbf{b} \in \mathbb{B}_{\mathbf{y}}$) is
\begin{equation}
\lambda_{\mathbf{y}}\left(\mathbf{b};\theta\right)=\prod_{t=1}^{T} \omega_{t} \left[F\left(\bar{b}_{t} \right)-F\left(\underline{b}_{t} \right)\right],\label{eq: lambda_evaluator}
\end{equation}
where the weights $\omega_{t}$ are as given in the statement of Algorithm \ref{alg: scenario-sampler-peer} and $\bar{b}_{t}$ and $\underline{b}_{t}$ are the bucket boundaries for agent $t$ defined earlier. Note that $\lambda_{\mathbf{y}}\left(\mathbf{b};\theta\right)$ is a ``re-weighted" version of $\zeta\left(\mathbf{b};\theta\right)$. We also have that $\lambda_{\mathbf{y}}\left(\mathbf{b};\theta\right)=0$ for all $\mathbf{b}\in\mathbb{B}\setminus\mathbb{B}_{\mathbf{y}}$.
\\
\begin{algorithm}
\caption{\label{alg: threshold-finder-peer}\textsc{Threshold Finder} (peer effects case) }

\textbf{Inputs:} $\mathbf{z}=\left(\mathbf{X},\mathbf{y}\right)$,
\textbf{$\theta$}, $\mathbf{U} $,
$t$
\begin{enumerate}
\item For $t'=1,\ldots,T$ 
\begin{enumerate}
\item if $y_{t'}=0$, then set $\tilde{U}_{t'}=U_{t'} $;
\item if $y_{t'}=1$, then
\begin{enumerate}
\item if $t'<t$, then set $\tilde{U}_{t'}=U_{t'}$
($h_{t'}$ already found);
\item if $t'>t$, then set $\tilde{U}_{t'}=X_{t}'\beta-1$ ($h_{t'}$
not already found).
\end{enumerate}
\end{enumerate}
\item $\tilde{U}_{t}=X_{t}'\beta+\mathbf{G}_{t}\mathbf{y}\delta+1$
(ensures that player $t$ \emph{will not }want to choose $\tilde{Y}_{t}=1$
in Step 3 below).
\item Find the minimal NE, $\tilde{\mathbf{Y}}$, associated with $\tilde{\mathbf{U}}$.
Set $h_{t} =X_{t}'\beta+\mathbf{G}_{t}\tilde{\mathbf{Y}}\delta$.
\end{enumerate}
\textbf{Output:} The threshold, $h_{t}$.
\end{algorithm}
\\
A key step of Algorithm \ref{alg: scenario-sampler-peer} is its calling of our \textsc{Threshold Finder} (i.e., Algorithm \ref{alg: threshold-finder-peer}).
By construction the \textsc{Threshold Finder} is first called after all draws of $U_{t} $ for $y_{t}=0$ agents have already been made (see Step 2 of Algorithm \ref{alg: scenario-sampler-peer}). Hence, at the start of each call of Algorithm \ref{alg: threshold-finder-peer}, $\mathbf{U} $ consists of draws of $U_{t}$ for all $y_{t}=0$ cases, those draws of $U_{t}$ for the $y_{t}=1$ cases for which the thresholds $h_{t}$ have already been determined, and the initialization values for any remaining elements.
\\
\\
Let $t$ index the player for which $h_{t}  $ is currently being determined and $t'$ other players. In Step 1.a.ii. of the \textsc{Threshold Finder} we \emph{provisionally} set the random taste shock for all $t'$ where $y_{t'}=1$ \emph{and} $h_{t'}$ has not yet been determined, such that in the minimal equilibrium computed in Step 3 we will have $\tilde{y}_{t'}=1$ with probability one. For all other players, except the current one, the relevant random taste shocks have already been chosen.
\\
\\
In Step 2 of Algorithm \ref{alg: threshold-finder-peer} we then (provisionally) set the current player's random utility draw, $\tilde{U}_{t}$, to a level such that the strictly dominant strategy will be for player $t$ to \emph{not} to take action \emph{even when all other player actions are as
specified in the target NE $\mathbf{y}$}. In the target NE $\mathbf{y}$ we have $y_{t}=1$ (Algorithm \ref{alg: threshold-finder-peer} is only called in such cases), but here we wish to induce a different NE $\tilde{\mathbf{Y}}\preceq\mathbf{y}$ with $\tilde{Y}_{t}=0$. By forcing player $t$ to \emph{not} take action it may be that we also induce some other players whose thresholds $h_{t'}$ have already been chosen to also \emph{not} take the action as well, such that $\tilde{Y}_{t'}=0$ (even though $y_{t'}=1$ in the target NE). This provides useful information since it helps us understand how important player $t$'s decision to take the action is for sustaining the target NE.
\\
\\
Observe that, by monotonicity of the utility function of the underlying game we have that $\tilde{\mathbf{Y}}\preceq\mathbf{y}$: the equilibrium simulated
in Algorithm \ref{alg: threshold-finder-peer} is less dense than the target one. In Step 3 of Algorithm \ref{alg: threshold-finder-peer} we set $h_{t} =X_{t}'\beta+\mathbf{G}_{t}\tilde{\mathbf{Y}}\delta$. This is exactly the threshold we need. We know that if $U_{t}  \leq h_{t}  $ it (i) will be optimal for player $t$ to take action given the actions of her peers, $\tilde{\mathbf{Y}}_{-t}$, in the sparser equilibrium and (ii) because player $t$ will take the action, no other players will want to deviate from the target NE, $\mathbf{y}$ to $\tilde{\mathbf{Y}}$. Hence agents will choose $\mathbf{y}$ as desired. However if $U_{t} > h_{t}$, player $t$ will choose not to take the action and other players will also not take further actions beyond $\tilde{\mathbf{Y}}$. For such a draw of $\mathbf{U}$, $\mathbf{y}$ will not be the NE. 
\\
\\
These observations are formalized in the following Theorem.
\begin{thm}
\label{thm: simulation-peer}(\textsc{Valid Sampler - Peer Effects Game}) Consider the complete information $T$-player peer effects game defined above. For any minimal NE,  $\mathbf{y}\in\mathbb{Y}^{T}$, 
Algorithm \ref{alg: scenario-sampler-peer}, in conjunction with the \textsc{Threshold Finder} (Algorithm \ref{alg: threshold-finder-peer}) generates a random $\mathbf{U}  \in\mathbf{b}$ such that (i) $\mathbf{b}\in\mathbb{B}_{\mathbf{y}}$ with probability one and (ii) $\lambda_{\mathbf{y}}\left(\mathbf{b};\theta\right)>0$ for all $\mathbf{b}\in\mathbb{B}_{\mathbf{y}}$.
\end{thm}
\begin{proof}
See Appendix \ref{app: Proof-of-Theorem}.
\end{proof}

A key feature of Theorem \ref{thm: simulation} is the guarantee not just that $\mathbf{U}  \in\mathbf{b}$ with $\mathbf{b}\in\mathbb{B}_{\mathbf{y}}$, but that all scenarios in $\mathbb{B}_{\mathbf{y}}$ are visited with positive probability. Theorem \ref{thm: simulation} is sufficient for consistency of \eqref{eq: importance_likelihood_simulator} for the true likelihood as $S\rightarrow\infty$.

\subsection{The coordination game}

In this section we formally walk through our scenario sampling algorithm using the simple two-player coordination game introduced in the previous section. Our goal is to construct a simulation estimate of the ex ante probability that $\mathbf{Y}= {1 \choose 1}$ for a given $\theta$. There are three scenarios which lead to the observed outcome; these are the darker shaded lower left-hand-side rectangles in Figure \ref{fig: understanding-scenarios}. We now illustrate how the \textsc{Scenario Sampler}, Algorithm \ref{alg: scenario-sampler-peer}, samples each of them with positive probability. For the purposes of illustration, and to contrast with our informal analysis in the prior section, will will draw Ademaro's preference shock $U_1$ first, followed by Brunhilde's, $U_2$.
\\ 
\\ 
\textbf{Sampling scenario $\mathbf{b}_{11}$: Strictly dominant for both Ademaro and Brunhilde to attend}
\\
\\ 
Step 1 is initialization. Step 2 is not relevant for the $\mathbf{Y}= {1 \choose 1}$ NE.
\\
\\ 
In the first execution of Step 3, the \textsc{Threshold Finder} proceeds under the presumption that in the end Brunhilde will want to choose $Y_2=1$ (see Step 1, (b) ii. in \textsc{Threshold Finder}). It therefore finds a threshold for Ademaro of $h_1 = X_{1}'\beta^{\left(0\right)}+\delta^{\left(0\right)}$.
The \textsc{Scenario Sampler} consequently draws $U_1$ from $\left(-\infty, {h}_{1}  \right]$ (see Step 3, (a) ii. in \textsc{Scenario Sampler}). In order eventually land in $\mathbf{b}_{11}$, it must be the case that our $U_1$ draw is lower then $X_{1}'\beta^{\left(0\right)}$.
\\
\\
In the second execution of Step 3, the \textsc{Threshold Finder} knows that Ademaro will choose $y_t=1$ (Step 1, (b) i. in \textsc{Threshold Finder}) and therefore finds a threshold of $h_2 = X_{2}'\beta^{\left(0\right)}+\delta^{\left(0\right)}$ for Brunhilde.
The $U_2$ is drawn from $\left(-\infty,{h}_{2}  \right]$. In order to land in $\mathbf{b}_{11}$ it must be the case that this draw is below $X_{2}'\beta^{\left(0\right)}$.
\\
\\
By using the truncated probability distributions given in the statement of the algorithm, we find that $\mathbf{b}_{11}$ is drawn with an ex ante probability $\lambda_{\mathbf{y}}\left(\mathbf{b}_{21};\theta^{\left(0\right)}\right)	= \frac{F\left(X_{1}'\beta^{\left(0\right)}\right)}{F\left(X_{1}'\beta^{\left(0\right)}+\delta^{\left(0\right)}\right)}\left[\frac{F\left(X_{2}'\beta^{\left(0\right)}\right)}{F\left(X_{2}'\beta^{\left(0\right)}+\delta^{\left(0\right)}\right)}\right].$\\
\\
\textbf{Sampling scenario $\mathbf{b}_{21}$: Ademaro is ``on the fence"}
\\
\\
Step 1 is initialization. Step 2 is not relevant for the $\mathbf{Y}= {1 \choose 1}$ NE.
\\
\\
In the first execution of Step 3, the \textsc{Threshold Finder} again proceeds under the presumption that in the end Brunhilde will want to choose $Y_2=1$ (see Step 1, (b) ii. in \textsc{Threshold Finder}). It therefore again finds a threshold for Ademaro of $h_1 = X_{1}'\beta^{\left(0\right)}+\delta^{\left(0\right)}$. As above, the \textsc{Scenario Sampler} draws $U_1$ from $\left(-\infty, {h}_{1}  \right]$ (see Step 3, (a) ii. in \textsc{Scenario Sampler}). However in this instance we want to end up in scenario $\mathbf{b}_{21}$. Therefore it must be the case that our $U_1$ draw is in the interval $\left(X_{1}'\beta^{\left(0\right)},X_{1}'\beta^{\left(0\right)}+\delta^{\left(0\right)} \right)$.
\\
\\
In the second execution of Step 3, the \textsc{Threshold finder} knows that \textit{Ademaro will only want to take the action if Brunhilde takes the action} (Step 1, (b) i. in \textsc{Threshold Finder}). Therefore it finds a lower threshold for Brunhilde of $h_2 = X_{2}'\beta^{\left(0\right)}$. The $U_2$ preference shock is drawn from $\left(-\infty,{h}_{2}  \right]$. In order to land in $\mathbf{b}_{21}$ it must be the case that this draw is lower than $X_{2}'\beta^{\left(0\right)}$ (which it is by construction).
\\
\\
$\mathbf{b}_{21}$ is drawn with an ex ante probability of $\lambda_{\mathbf{y}}\left(\mathbf{b}_{21};\theta^{\left(0\right)}\right)	= \frac{F\left(X_{1}'\beta^{\left(0\right)}+\delta^{\left(0\right)}\right)-F\left(X_{1}'\beta^{\left(0\right)}\right)}{F\left(X_{1}'\beta^{\left(0\right)}+\delta^{\left(0\right)}\right)}.$
\\
\\
\textbf{Sampling scenario $\mathbf{b}_{12}$: Brunhilde is ``on the fence"}
\\
\\
This scenario is similar to the $\mathbf{b}_{11}$ scenario. It differs only in terms of the realized draw of $U_2$. The ex ante probability of drawing this scenario is $\lambda_{\mathbf{y}}\left(\mathbf{b}_{12};\theta^{\left(0\right)}\right)	= \frac{F\left(X_{1}'\beta^{\left(0\right)}\right)}{F\left(X_{1}'\beta^{\left(0\right)}+\delta^{\left(0\right)}\right)}\left[\frac{F\left(X_{2}'\beta^{\left(0\right)}+\delta^{\left(0\right)}\right)-F\left(X_{2}'\beta^{\left(0\right)}\right)}{F\left(X_{2}'\beta^{\left(0\right)}+\delta^{\left(0\right)}\right)}\right]$. 
\\
\\
Observe how in all three cases, the \textsc{Threshold Finder} finds the appropriate threshold for Brunhilde by using information contained in Ademaro's $U_1$ draw. Specifically the algorithm determines whether Ademaro is sensitive to Brunhilde's play, if he is it chooses a lower threshold for Brunhilde' shock to ensure that, in the end, both of them will want to go to the concert.
\\
\\
In the example above we construct $\mathbf{U}  \in\mathbf{b}$ with $\mathbf{b}\in\mathbb{B}_{\mathbf{y}}$ by drawing Ademaro's preference shock first, followed by Brunhilde's. In our discussion of the same example in the prior section we reversed this order. By comparing the scenario probabilities above with those calculated earlier (see Equation \eqref{eq: lambdas_coordination_ex1}) we can see that the order in which we draw the shocks for the $y_t=1$ agents matters.
\\
\\
We speculate that there is an optimal ordering of the $y_t=1$ agents. That is an ordering which will result in importance sample weights that estimate the likelihood with the smallest amount of simulation error for a fixed number of simulation draws. Our conjecture is that the optimal ordering involves a sorting of agents by the linear indices $X_{t}'\beta^{\left(0\right)}$. We leave a complete analysis of this question to future work.

\subsection{Technical comments}
In this section we briefly discuss a few additional feature of our algorithm. Additional details are provided in the appendices.
\\
\\
Section \ref{sec: discrete supermodular games} below introduces a class of binary-action, complete information, supermodular games to which an extended version of our scenario sampler can be applied. In these games $T$ players each take $M$ binary actions. For most games in this class the number of scenarios grows exponentially with $TM$. At the same time the probability of observing any particular NE generally shrinks exponentially. This makes explicit enumeration of scenarios, required for ``direct" MLE, infeasible. This feature of our problem also makes crude frequency-based -- ``dartboard" -- Monte Carlo methods completely impractical outside of toy examples.
\\
\\
In contrast, our approach allows a researcher to randomly sample a scenario from the target set $\mathbb{B}_{\mathbf{y}}$ in polynomial time. This is the main reason why our approach makes SML estimation of games with several thousands strategic decision possible.  However there are a few other properties of our procedure which make scenario-based estimation even more applicable to larger games. 
\\
\\
\textbf{Differentiability}
\\
An advantage of importance sampling, relative to crude frequency-based Monte Carlo, is that the former results in a criterion function that is differentiable in $\theta$, while the latter does not (see \cite{McFadden_EM89, Ackerberg_QME2009}. Inspection of Equation \eqref{eq: importance_likelihood_simulator} reveals that is locally differentiable in $\theta$. The details of this claim are spelled out in Appendix \ref{app: differenication}.
\\
\\
When agents' preferences are indexed by many parameters, it is very helpful to have the derivative of the log-likelihood function with respect to $\theta$ available. Non-gradient based optimization approaches do not scale well to high-dimensional settings. In our empirical illustration we fit a network formation model with over $200$ parameters by SML, quasi-Newton assisted, estimation (the model includes household-specific degree heterogeneity parameters as in \cite{Graham_EM17, Graham_HBE2020}). Fitting a model of this size would not be practical without gradient-based optimization methods.
\\
\\
\textbf{Scenario Recycling}
\\
Maximizing the log-likelihood requires it to be evaluated at many different values of $\theta$. The most expensive component of log-likelihood evaluation involves computation of the Nash Equilibrium (NE). This observation applies not just to our scenario-based approach, but to other likelihood-based methods of game estimation \citep[e.g.,][]{Bajari_et_al_EM10}. In Appendix \ref{app: recycling} we outline a procedure, ``scenario recycling", which economizes on repeated NE calculations. Our procedure is related to the methods of \cite{Ackerberg_QME2009} as applied to games by \cite{Bajari_et_al_EM10}. As with these approaches, scenario recycling works only if there is exactly one strategic parameter to estimate. This applies to the peer effect example of this section and the network formation example we develop empirically.
\\
\\
\textbf{Independence}\\
In many applications a researcher will observe the NEs of several, independent, games. For example a researcher may observe the smoking behavior of adolescents across many different secondary schools. Such applications have a number of computational and statistical advantages.
\\
\\
In these applications the sample log likelihood is the sum of the log-likelihoods associated with each observed game. These log-likelihood components can be estimated separately via their own simulated scenario samples and then aggregated. 
\\
\\
Independence across games also allows for the application of textbook SML large sample theory \citep[e.g.,][]{Newey_McFadden_HBE94, Hajivassiliou_Ruud_HBE1994}. Our approach can also be used to analyze a single large game, but such applications raise novel statistical issues which we do not address here. Some additional discussion on this point is provided when we discuss our empirical illustration.

\section{\label{sec: discrete supermodular games}Supermodular games}

In this section show how to adapt Algorithm's \ref{alg: scenario-sampler-peer} and \ref{alg: threshold-finder-peer} to a broader class of binary-action supermodular games. In this class of games each of $T$-players decides whether to take, or not to take, each of $M$ (non mutually exclusive) binary actions. In the peer effects game analyzed in the previous section $M=1$; but in many games agents may make multiple strategic decisions. In a game of directed network formation, for example, each player decides whether to direct a link (or not) to each of the $T-1$ other players in the network (such that $M=T-1$).  
\\
\\
In order to present positive results, we need to restrict the structure of payoffs to ensure that the resulting game is supermodular. Let $Y_{tm}\in\left\{ 0,1\right\} $ denotes the $m^{th}$ action
of player $t$. Denote agent $t$'s full $M \times 1$ action vector by $\mathbf{Y}_t=\left(Y_{t1},\ldots,Y_{tM}\right)'$.  Let $\mathbf{Y}_{t,-m}=\left(Y_{t1},\ldots,Y_{tm-1},Y_{tm+1},\ldots,Y_{tM}\right)'$
denote player $t$'s $M-1$ actions other than $Y_{tm}.$ Similarly
let $\mathbf{Y}_{-t}=\left(\mathbf{Y}_{1}',\ldots,\mathbf{Y}_{t-1}',\mathbf{Y}_{t+1}',\ldots,\mathbf{Y}_{T}'\right)'$
be the $\left(T-1\right)M\times1$ vector of actions taken by agents
other than $t$, henceforth called her \emph{peers}. Let $\mathbf{U}$
and $\mathbf{X}$ be matrices consisting of all taste shocks
and covariates. A player's realized utility depends on their own actions, $\mathbf{y}_{t} \in \mathbb{Y}_t  \overset{def}{\equiv} \{0,1\}^{M},$
as well as the actions of their \emph{peers}, $\mathbf{y}_{-t}$. For
pure strategy profile $\mathbf{y}\in \mathbb{Y} \overset{def}{\equiv} \{0,1\}^{TM}$ the utility function
of player $t$, $\upsilon_{t}: \mathbb{Y} \rightarrow\mathbb{R}$ is
\begin{align}\label{eq: general_utility}
\upsilon_{t}\left(\mathbf{y};\mathbf{X},\mathbf{U},\theta\right)=&\upsilon\left(\mathbf{y}_{t},\mathbf{y}_{-t};X{}_{t},\mathbf{U}_{t},\theta\right)\\\overset{def}{\equiv}&\sum_{m=1}^{M}y_{tm}\left(X_{tm}'\beta_{m}+s_{m}\left(\mathbf{y}_{t,-m},\mathbf{y}_{-t}\right)'\delta_{m}-U_{tm}\right),\nonumber 
\end{align}
with $\theta=\left(\beta_{1}',\delta_{1}',\ldots,\beta_{M}',\delta_{M}'\right)'$ and where $s_{m}\left(\mathbf{y}_{t,-m},\mathbf{y}_{-t}\right)$ is a given vector-valued function of own actions (other
than the $m^{th}$ one) and peers' actions. This function
may vary with $m=1,\ldots,M$. It could also depend on exogenous agent-by-action covariates, but we suppress this in the notation. The setup also allows for player-by-action specific covariates, $X_{tm}$, to influence payoffs. Players choose actions to maximize
their utility given the actions of their peers under complete information
(i.e., agents best respond).
\\
\\
The effect of an \emph{increase }in $U_{tm}$ is to \emph{decrease}
player $t$'s marginal benefit of taking action $m$; it can be thought
of as a distaste or cost-of-action shock. This term generates unobserved agent-specific
heterogeneity in the marginal utilities attached to taking the $M$
actions. This $M$-vector endows our model with the classic random
utility structure pioneered by \citet{McFadden_FinE74}. We assume
that the elements of $\mathbf{U}_{t}=\left(U_{t1},\ldots,U_{tM}\right)'$ are
independently and identically distributed (iid) with known cumulative
distribution function (CDF) $F\left(\cdot\right)$ and probability
density function (PDF) $f\left(\cdot\right)$. Independence is also maintained across agents and games (the $t$ and $i$ subscripts). We briefly discuss
different distributional and dependence assumptions on $\mathbf{U}_{t}$
in the conclusion.
\\
\\
The $s_{m}\left(\mathbf{y}_{t,-m},\mathbf{y}_{-t}\right)$ term allows for player $t$'s marginal
utility from action $m$ to depend on what other actions she chooses
to take.  For example the payoff from smoking may vary with whether she also decides to drink. This term also captures how the choices
of other players in the game alter the utility player $t$ attaches
to action $m$; so called \emph{endogenous effects} in the parlance
of \citet{Manski_ReStud93}.\footnote{\emph{Exogenous} or \emph{contextual effects} can be added to (\ref{eq: general_utility})
simply by defining $x_{tm}$ to include, for example, averages of the
attributes of other players in game $i$. Here our focus is on the
implications of strategic interaction for estimation and inference,
consequently we abstract from exogenous effects. }
\\
\\
To ensure the resulting game is supermodular, we require that $s_{m}\left(\mathbf{y}_{t,-m},\mathbf{y}_{-t}\right)$ is monotone increasing in both its arguments and that $\delta_{m} \geq 0$. This restriction ensures that (i) own actions are (weak) complements with one another and that (ii) own and peer actions are weakly complementary as well.
\\
\\
The first claim follows from the restrictions that (i) the elements of
$s_{m}\left(\mathbf{y}_{t,-m},\mathbf{y}_{-t}\right)$ are weakly increasing in $\mathbf{y}_{t,-m}$ and that (ii) the elements of $\delta_{m}$ are non-negative for $m=1,\ldots,M$. Observe that $\left(\mathbb{Y}_{t},\succeq\right)$
is a complete lattice and, further, that for all $\mathbf{y}_{t},\mathbf{y}_{t}'\in\mathbb{Y}_{t}$
we have that
\begin{align} \label{eq: supermodularity}
\upsilon\left(\mathbf{y}_{t}\lor \mathbf{y}_{t}',\mathbf{y}_{-t};x{}_{t},\mathbf{u}_{t},\theta\right) + & \upsilon\left(\mathbf{y}_{t}\land \mathbf{y}_{t}',\mathbf{y}_{-t};x{}_{t},\mathbf{u}_{t},\theta\right) \\ \geq & \upsilon\left(\mathbf{y}_{t},\mathbf{y}_{-t};x{}_{t},\mathbf{u}_{t},\theta\right)+\upsilon\left(\mathbf{y}_{t}',\mathbf{y}_{-t};x{}_{t},\mathbf{u}_{t},\theta\right), \notag
\end{align}
and hence that $\upsilon\left(\mathbf{y}_{t},\mathbf{y}_{-t};x{}_{t},\mathbf{u}_{t},\theta\right)$
is a\emph{ supermodular} function of $\mathbf{y}_{t}$ \citep{Topkis_Book98}.\footnote{The notation $\mathbf{y}_{t}\lor \mathbf{y}_{t}'$ denotes the join or least upper bound operation,
$\mathbf{y}_{t}\land \mathbf{y}_{t}'$ the meet or greatest lower bound operation.} No two actions a player can take are substitutes for one another.
\\
\\
The second restriction on $s_{m}\left(\mathbf{y}_{t,-m},\mathbf{y}_{-t}\right)$ also has a nice economic interpretation. It implies that agent $t$'s utility from taking binary actions $m=1,\ldots,M$
is weakly increasing in $\mathbf{y}_{-t}$. Noting that $\left(\mathbb{Y}\setminus\mathbb{Y}_{t},\succeq\right)$
is also a complete lattice, we have that, for all $\mathbf{y}_{t}'\succeq \mathbf{y}_{t}$
and $\mathbf{y}_{-t}'\succeq\mathbf{y}_{-t}$, the \emph{increasing
differences} property
\begin{equation}
\upsilon\left(\mathbf{y}_{t}',\mathbf{y}_{-t}';x{}_{t},\mathbf{u}_{t},\theta\right)-\upsilon\left(\mathbf{y}_{t},\mathbf{y}_{-t}';x{}_{t},\mathbf{u}_{t},\theta\right)\geq\upsilon\left(\mathbf{y}_{t},\mathbf{y}_{-t}';x{}_{t},\mathbf{u}_{t},\theta\right)-\upsilon\left(\mathbf{y}_{t},\mathbf{y}_{-t};x{}_{t},\mathbf{u}_{t},\theta\right).\label{eq: increasing_differences}
\end{equation}
Own and peer actions are complementary.
\\
\\
An implication of restrictions \eqref{eq: supermodularity} and \eqref{eq: increasing_differences} is that the game is supermodular in the sense of \citet{Milgrom_Roberts_EM1990}. They show that in supermodular
games there exist two extremal NE in pure strategies -- minimal and
maximal -- and that all rationalizable strategy profiles are bounded
by these two extremal NE. The monotonicity of the utility function in $\mathbf{y}$ further allows us to find the minimal NE using Tarski's \citeyearpar{Tarski_PJM55}
Theorem. Our algorithm exploits these implications of supermodularity.
\\
\\
Our assumptions about sampling and the data generating process are
collected in Assumption \ref{ass: supermodular_game}.
\begin{assumption}
\label{ass: supermodular_game}\textsc{(Supermodular Game) } (i) The payoff function (\ref{eq: general_utility}) satisfies restrictions \eqref{eq: supermodularity} and \eqref{eq: increasing_differences};
(ii) the elements of $\mathbf{U}_{t}=\left(U_{t1},\ldots,U_{tM}\right)'$
are iid with known CDF $F\left(\cdot\right)$; (iii) agents choose
actions under complete information (i.e., they know the structure of preferences as well as $\mathbf{X}$ and $\mathbf{U}$) and (iv) play the minimal NE when multiplicity is present.
\end{assumption}

\subsection*{Examples }

Although Assumption \ref{ass: supermodular_game} is a real restriction,
it is sufficiently flexible to accommodate
many complete information games of interest to economists. To give some sense of the range of possible applications of our
methods, it is helpful to consider a few examples. Our first example is closely related to the peer effects model analyzed in the previous section.
\begin{example}
\label{ex: network_effects} \textsc{(Network externalities).} A leading
application of the methods outlined is this paper is to the study
technology adoption in the presence of network externalities \citep[e.g., ][]{Ackerberg_Gowrisankaran_RAND2006}.
Consider a single binary decision of whether to adopt $Y_{t}=1$ an
innovation or not $Y_{t}=0$. For example, \citet{Goolsbee_Klenow_JLE2002}
study network externalities in home computer adoption. If the marginal
benefits of adoption are increasing in the aggregate number of adopters,
then utility (\ref{eq: general_utility}) might take the form
\begin{equation}
\upsilon\left(y_{t},\mathbf{y}_{-t};x_{t},u_{t},\theta\right)=y_{t}\left(x_{t}'\beta+\left[\sum_{s\neq t}y_{s}\right]\delta-u_{t}\right)\label{eq: utility_network_effects}
\end{equation}
where, since $M=1$, we set $U_{t1}=U_{t}$ to simplify the notation;
similarly $s_{1}\left(\mathbf{y}_{-t}\right)=s\left(\mathbf{y}_{-t}\right)=\sum_{s\neq t}y_{s}$
(i.e., when $M=1$ we drop the $m$ subscript on the ``strategic" utility term).
\end{example}
Our second example nests a model analyzed by \citet{Krauth_JOE2006}
and \citet{Soetevent_Kooreman_JAE07}, itself a complete information
version of the seminal binary action peer effects model introduced by \citet{Brock_Durlauf_RES01,Brock_Durlauf_HBE01}. 
\begin{example}
\label{ex: peer_effects}\textsc{(Multi-Action Peer Effects) }Let
$Y_{t1}$ and $Y_{t2}$ be a pair of binary actions plausibly subject
to peer influence. For example smoking and drinking among adolescents
\citep[e.g., ][]{Gaviria_Raphael_RESTAT2001}. Here utility might
take the form
\begin{align*}
\upsilon\left(\mathbf{y}_{t},\mathbf{y}_{-t};x_{t},\mathbf{u}_{t},\theta\right)= & y_{t1}\left(x_{t}'\beta_{1}+y_{t2}\delta_{11}+\left[\frac{1}{T-1}\sum_{s\neq t}y_{s1}\right]\delta_{12}+\left[\frac{1}{T-1}\sum_{s\neq t}y_{s2}\right]\delta_{13}-u_{t1}\right)\\
 & +y_{t2}\left(x_{t}'\beta_{2}+y_{t1}\delta_{21}+\left[\frac{1}{T-1}\sum_{s\neq t}y_{s1}\right]\delta_{22}+\left[\frac{1}{T-1}\sum_{s\neq t}y_{s2}\right]\delta_{23}-u_{t2}\right).
\end{align*} Here $\ensuremath{s_{1}\left(\mathbf{y}_{t,-1},\mathbf{y}_{-t}\right)}=\left(y_{t2},\frac{1}{T-1}\sum_{s\neq t}y_{s1},\frac{1}{T-1}\sum_{s\neq t}y_{s2}\right)'$ and $\ensuremath{s_{2}\left(\mathbf{y}_{t,-2},\mathbf{y}_{-t}\right)}=\left(y_{t1},\frac{1}{T-1}\sum_{s\neq t}y_{s1},\frac{1}{T-1}\sum_{s\neq t}y_{s2}\right)'$.
\\
\\
In this model smoking and drinking are complementary.
Similarly the utility of smoking and drinking is increasing in the
fraction of peers that also engage in these behaviors. Our approach
to estimation scales especially well in this example, easily accommodating
games with, for example, hundreds of players (such that $2^{MT}$,
the number of possible pure strategy combinations, is very large).
\end{example}
A third example, related to the first two, corresponds to the setting
considered by \citet{Sundararajan_BE2008}, \citet{Banerjee_et_al_Sci13},
\citet{Kim_et_al_Lancet2015} and others.
\begin{example}
\label{ex: networked_peers}\textsc{(Adoption games on networks).
}In this setting the influence of individuals on each other is mediated
by an exogenously given network of relationships. As in Example \ref{ex: network_effects},
agents decide whether to adopt or not, but now the utility of adoption
for agent $t$ only varies with the adoption behavior of those agents
to which she is directly connected. Let $\mathbf{D}=\left[D_{st}\right]_{s,t=1}^{T}$
be a $T\times T$ binary adjacency matrix describing the structure
of links among the $T$ players in a game. Utility equals 
\[
\upsilon\left(y_{t},\mathbf{y}_{-t};x_{t},u_{t},\theta\right)=y_{t}\left(x_{t}'\beta+\left[\sum_{s\neq t}D_{ts}y_{s}\right]\delta-u_{t}\right)
\]
such that $s\left(\mathbf{y}_{-t}\right)=\sum_{s\neq t}D_{ts}y_{s}$.
Researchers have been especiallly interested in how the form of $\mathbf{D}$
-- the network -- shapes equilibrium adoption decisions. Related
is the question of how to allocate adoption subsidies, here conceptualized
as external manipulations of $X_{t}$, to maximize aggregate take-up.
Both \citet{Banerjee_et_al_Sci13} and \citet{Kim_et_al_Lancet2015}
present evidence suggesting that (scarce) subsidies should be allocated
toward more central players in a network. 
\end{example}
Our fourth example is adapted from \citet{Miyauchi_JOE16}, who pointed
out the connection between the theory of supermodular games and some
models of strategic network formation.\footnote{\citet{Miyauchi_JOE16} considered undirected networks, while our
analysis formally pertains to directed ones.}
\begin{example}
\label{ex: network_formation} \textsc{(Strategic network formation).
}In this example each of the $t=1,\ldots,T$ agents in a game decides
whether to direct a link, $Y_{ts}=1$, or not, $Y_{ts}=0$, to each
of the $T-1$ other agents $s\neq t$, $s=1,\ldots,T$ also in the
game \citep[e.g.,][]{Bala_Goyal_EM00,dePaula_et_al_EM18,Sheng_EM20,Pelican_Graham_NBER2020}.
Here $M$ -- the number of strategic decisions each players makes
-- equals $T-1$ -- the number peers to which an agent may direct
links. A pure strategy combination consists of a total of $T\left(T-1\right)$
binary decisions for a total of $2^{T\left(T-1\right)}$ possible
directed network configurations. In this example it is convenient
to slightly re-define $\mathbf{Y}$ to be the $T\times T$ digraph
adjacency matrix $\text{\ensuremath{\left[Y_{ts}\right]}}_{s,t=1}^{T}$
with $Y_{tt}\equiv0$ for $t=1,\ldots,T$ by construction \citep[cf.,][]{Graham_HBE2020}.
Even in modestly-sized networks, consisting of, say, $T\approx100$
agents, this is a very large problem relative to extant applications
of game theory in econometrics (a total of $2^{9900}$ possible pure
strategy combinations!). If agents prefer, say, reciprocation and transivity
in links \citep[e.g., ][]{deWeerdt_IAP04}, then utility for agent
$t$ in a network formation game might take the form
\[
\upsilon\left(\mathbf{y}_{t},\mathbf{y}_{-t};x_{t},\mathbf{u}_{t},\theta\right)=\sum_{s\neq t}y_{ts}\left(x_{ts}'\beta+y_{st}\delta_{1}+\left[\sum_{r=1}^{T}y_{tr}y_{rs}\right]'\delta_{2}-u_{ts}\right),
\]
with $\delta_{1}$ and $\delta_{2}$ respectively measuring the strength
of agents' taste for reciprocity and transitivity in links and $x_{ts}$
consisting of a vector of sending ($t$ subscript) and receiving ($s$
subscript) agent-specific regressors as well as possibly dyadic regressors
(e.g., distance between $t$ and $s$). Of special interest is the
case where only a single network is observed (i.e., ``the fixed-$N$,
large-$T$, large-$M$ case''). Of course we might also be interested
in settings where agents may form multiple types of links. Such an
extension would incorporate features of Example \ref{ex: peer_effects}. Note that in this example $\ensuremath{s\left(\mathbf{y}_{t,-s},\mathbf{y}_{-t}\right)}=\left(y_{st},\sum_{r=1}^{T}y_{tr}y_{rs}\right)'$ and $\delta=\left(\delta_{1},\delta_{2}\right)'$, both of which do not vary with $s=1,\ldots,t-1,t+1,\ldots,T$ (an appropriate restriction when agents are conditionally exchangeable given covariates).
\end{example}
Our final example, due to \citet{Jia_EM08} and \citet{Nishida_MS2015},
is well-known in the field of empirical industrial organization.
\begin{example}
\label{ex: multi_mkt_entry} \textsc{(Strategic multi-market entry).}
Let $m=1,\ldots,M$ index markets that two rival firms $t=1,2$ choose
to either enter, $Y_{tm}=1$, or not, $Y_{tm}=0$. Let $z_{ml}$ denote
the distance between markets $m$ and $l$ and, as in \citet{Jia_EM08},
assume that the benefits of entry into market $m$ are increasing
in the number of nearby markets $l$ in which firm $t$ also operates,
specifically in $\frac{1}{z_{ml}}\sum_{l\neq m}y_{tl}$. We also assume
that the benefits of market entry are lower if one's rival also enters
the market:
\begin{equation}
\upsilon\left(\mathbf{y}_{t},\mathbf{y}_{-t};x_{t},\mathbf{u}_{t},\theta\right)=\sum_{m=1}^{M}y_{tm}\left(x_{tm}'\beta_{t}+\delta_{t1}\left[\frac{\sum_{l\neq m}y_{tl}}{z_{ml}}\right]+\delta_{t2}y_{-tm}-u_{tm}\right).\label{eq: entry_utility}
\end{equation}
Utility (\ref{eq: entry_utility}) captures firms' preferences to
operate in spatially clustered markets in order to, for example, economize
on the number of distribution centers they need to operate ($\delta_{t1}\geq0$).
It also captures the fact that firms prefer markets with less competition
($\delta_{t2}\leq0$). Of special interest is the case where only a
single game is observed, but the two firms potentially operate/compete
in many markets (i.e., ``the fixed-$N$, fixed-$T$, large-$M$ case'').
Although utility (\ref{eq: entry_utility}) is not weakly increasing
in $\mathbf{y}_{-t}$, if we redefine the strategy space for, say
player $2$, such that $Y_{2m}=0$ corresponds to entry into market
$m$ and $Y_{2m}=1$ to non-entry, then the non-decreasing property
is restored with $\delta_{11}\geq0$, $\delta_{12}\geq0$ and $\delta_{21}\leq0$
and $\delta_{22}\leq0$. See \citet{Jia_EM08} for details. Note this
example departs slightly from the general framework of \eqref{eq: general_utility}.
Here, in keeping with the fixed-$T$, large-$M$ framing, the utility
parameters are agent-specific (e.g., Walmart and Kmart may have different profit functions), \emph{not }action-specific (as is appropriate
for the case where $M$ is fixed, but agents are exchangeable, which
characterizes the hitherto introduced examples).
\end{example}

\subsection*{Scenario estimation for supermodular games}
Our scenario sampling algorithm is easily adapted to handle the more complicated games of this section. Denote the systematic component of player $t$'s action $m$ utility by
\begin{align}
g_{tm}\left(\mathbf{y}, x_{tm}\right) 
 & \overset{def}{\equiv}x_{tm}'\beta_{m}+s_{m}\left(\mathbf{y}_{t,-m},\mathbf{y}_{-t}\right)'\delta_{m}.
\end{align}
Minimal adjustment to the peer effect algorithm outlined earlier gives our general \textsc{Scenario sampler}.
\\
\begin{algorithm}[H]
\caption{\label{alg: scenario-sampler}\textsc{Scenario sampler}  (general case)}
\textbf{Inputs:} $\mathbf{z}=\left(\mathbf{x},\mathbf{y}\right)$, \textbf{$\theta$} (i.e., a target pure strategy
combination and a utility/payoff function).
\begin{enumerate}
\item Initialize $\mathbf{U}=\left(\mathbf{U}_{1},\ldots,\mathbf{U}_{T} \right)'=\begin{cases}
-\infty & \text{if } y_t = 1 \\
+\infty & \text{otherwise}
\end{cases}$.
\item For $t=1,\ldots,T$ and  $m= 1,\ldots,M$
\begin{enumerate}
\item If $y_{tm}=0,$ then sample $U_{tm} \in \left( g_{tm}\left(\mathbf{y}, x_{tm}\right) ,\infty\right)$
from the conditional density $\frac{f\left(u\right)}{1-F\left(g_{tm}\left(\mathbf{y}, x_{tm}\right)  \right)}\overset{def}{\equiv}\omega_{t}  f\left(u\right)$.
\end{enumerate}
\item For $t=1,\ldots,T$ and $m= 1,\ldots,M$
\begin{enumerate}
\item If $y_{tm}=1$, then 
\begin{enumerate}
\item determine $h_{tm}  $ using \textsc{Threshold}($\mathbf{z}$,\textbf{
$\theta$}, $\mathbf{U}  $, $t$,
$m$);
\item sample $U_{tm}  \in\left(-\infty, h_{tm}  \right]$
from the conditional density $\frac{f\left(u\right)}{F\left(h_{tm}  \right)}\overset{def}{\equiv}\omega_{tm}  f\left(u\right).$
\end{enumerate}
\item Find $\mathbf{b}\in\mathbb{B}_{\mathbf{y}}$ such that $\mathbf{U}  \in\mathbf{b}$.
\end{enumerate}
\end{enumerate}
\textbf{Outputs:} The $TM \times1$ weight vector $\underline{\omega}  =\left(\omega_{1}  ,\ldots,\omega_{TM} \right)'$,
the vector of taste shocks \textbf{$\mathbf{U}  $} and a (random) scenario $b\in\mathbb{B}_{\mathbf{y}}$.
\end{algorithm}

As in the peer effects special case, a key component of our general \textsc{Scenario sampler} is a \textsc{Threshold finder} subroutine.

\begin{algorithm}[H]
\caption{\label{alg: threshold-finder}\textsc{Threshold finder} (general case)}
\textbf{Inputs:} $\mathbf{z}=\left(\mathbf{X},\mathbf{y}\right)$,
\textbf{$\theta^{\left(0\right)}$}, $\mathbf{U}$,
$t$, $m$.
\begin{enumerate}
\item For $t'=1,\ldots,T$  and $m' = 1,\ldots,M$
\begin{enumerate}
\item if $y_{t'm'}=0$, then set $\tilde{U}_{t'm'}=U_{t'm'} $;
\item if $y_{t'm'}=1$, then
\begin{enumerate}
\item if $\bar{h}_{t'm'} $ already found, then set $\tilde{U}_{t'm'}=U_{t'm'}$;
\item otherwise, set $\tilde{U}_{t'm'}=g_{tm}\left(0, x_{tm}\right) -1$ (i.e., $\bar{h}_{t'm'} $
not already found).
\end{enumerate}
\end{enumerate}
\item Set $\tilde{U}_{tm}= g_{tm}\left(\mathbf{y}, x_{tm}\right) +1$
(this ensures that player $t$ \emph{will not} want to choose $\tilde{Y}_{tm}=1$
in Step 3 below).
\item Find the minimal NE, $\mathbf{\tilde{y}}$, associated with $\tilde{\mathbf{U}}$.
Set $h_{tm} = g_{tm}\left(\mathbf{\tilde{y}}, x_{tm}\right)$. 
\end{enumerate}
\textbf{Output:} The threshold, $h_{tm}$.
\end{algorithm}
The correctness of the general algorithm follows from a simple extension of Theorem \ref{thm: simulation-peer}.
\begin{thm}
\label{thm: simulation}(\textsc{Valid Sampler - General Case}) Consider the complete information $T$-player, $M$-action supermodular game defined above. For any minimal NE,  $\mathbf{y}\in\mathbb{Y}^{TM}$, 
Algorithm \ref{alg: scenario-sampler}, in conjunction with the \textsc{Threshold Finder} (Algorithm \ref{alg: threshold-finder}) generates a random $\mathbf{U}  \in\mathbf{b}$ such that (i) $\mathbf{b}\in\mathbb{B}_{\mathbf{y}}$ with probability one and (ii) $\lambda_{\mathbf{y}}\left(\mathbf{b};\theta\right)>0$ for all $\mathbf{b}\in\mathbb{B}_{\mathbf{y}}$.
\end{thm}
\begin{proof}
See Appendix \ref{app: Proof-of-Theorem}.
\end{proof}
The rest of estimation proceeds as described earlier. Although we note that scenario recycling is not straightforward when the strategic interaction parameter, $\delta$, is vector-valued (as discussed further in the Appendix \ref{app: recycling}). This does result in an increase in computation time; although since finding the minimal equilibrium is straightforward in supermodular games estimation remains feasible, even for very large games.

\section{\label{sec: monte carlo experiments}Monte Carlo Experiment}

In this section we summarize the results of a small number of Monte Carlo experiments. The purpose of the experiments is to verify our main theoretical claims as well as to get some sense of the small sample performance of our methods in an empirical setting of interest.
\\
\\
The Monte Carlo design uses a random geometric graph to construct a friendship network with $T \times T$ adjacency matrix $\mathbf{D}$ \citep[see][]{Graham_NBER16}. The friendship network is exogenous and determines each agent's set of peers. Specifically agents are scattered uniformly on a plane. The network is generated by randomly linking agents which are close to each other on this plane. The  purpose is to approximate a real world friendship network where only close agents have the possibility to meet each other.
\\
\\
The utility for agent $t$ has the form
\begin{equation*}
\upsilon\left(y_{t},\mathbf{y}_{-t};x_{t},u_{t},\theta\right)=y_{t}\left(x_{t}'\beta+\left[\sum_{s\neq t}D_{ts}y_{s}\right]\delta-u_{t}\right).
\end{equation*}
In the Monte Carlo design we use four covariates. Two of these covariates are binary (sampled for each player from a Bernoulli distribution). The remaining two covariates are continuously-valued (sampled from a uniform distribution). The preference shocks $\mathbf{U}=\left(U_1,\ldots,U_T\right)'$ are iid standard normal random variables. Full details of the data generating process are available in Appendix \ref{app: MC_details}. For each simulation we (i) draw the vector of utility shocks $\mathbf{U}$ and (ii) find the minimal equilibrium, $\mathbf{Y}$, by fixed point iteration. Finally, we estimate $\beta$ and $\delta$ based on the resulting observed $\mathbf{Y}$ vector by simulated maximum likelihood (SML) using our \textsc{Scenario sampler}. The regressor matrix $\mathbf{X}$ is simulated once and then held fixed across Monte Carlo replications.
\\
\\
We consider two variants of the above setup. In the first estimation is based upon the availability of many independent medium sized games. In this setting the asymptotic sampling distribution of the SMLE of $\theta$ follows from standard large sample results (see \cite{Hajivassiliou_Ruud_HBE1994, Newey_McFadden_HBE94}). This is a ``fixed $T$, fixed $M$, large $N$" setting.
\\
\\
We also consider the properties of our SML estimator when there is only a single large game. This is a ``large $T$, fixed $M$, fixed $N$" setting. Large sample theory for SML estimates is not available in this setting. Developing such theory raises a number of interesting questions that are well beyond the scope of this paper (see \cite{Menzel_ReStud2016} for some relevant ideas and also discussion in the next section).
\\
\\
The case of many medium sized games is encountered if the researcher has, for example, friendship data across many independent school classrooms \citep[e.g.,][]{vanRijsewijk_et_al_PLOS2018}. The case a single game arises when, for example, the researcher observes a network of relationships in a single village \citep[e.g.,][]{deWeerdt_IAP04}.  For the many medium sized games case we simulate datasets with $2,000$ agents belonging to one of $N=100$ separate friendship networks (each containing $T=20$ agents). For the single large game case we consider datsets with $T=500$ agents in $N=1$ friendship network.
\\ 
\\
Monte Carlo results for both cases are displayed in, respectively, Panels A and B of Table \ref{table:monte_carlo_results}. We report the average and standard deviation of the SML estimate of $\delta$ across $1000$ simulated datasets for each of the two designs. The true value of $\delta$ is $0.20$, which is close to the average of the SMLEs. Also reported is the size of a likelihood ratio test for $H_0: \delta=0.20$ and the coverage of a Wald-based 95 percent confidence interval for $\delta$. The standard errors used to construct this interval are based upon the simulated Hessian matrix (calculated by differentiating the simulated log-likelihood function). We each design we report results when the likelihood is estimate by drawing $S=1, 10$ and $100$ scenarios.
\\ 
\\
In both designs the SMLE of $\delta$ is approximately unbiased. This holds even when we use only a small number of scenario draws. However the normal approximation, as judged by the size of the LR test and the coverage of the confidence interval, only appears to be accurate for the many games design (consistent with extant large sample theory). This is confirmed by the histograms of the SMLEs for the two designs ($S=100$ cases) shown in Figure \ref{fig: mc-distribution}. The single game distribution in the right panel is notably skewed. Understanding the sampling properties of SMLEs in single large game settings (the ``large $T$, fixed $M$, fixed $N$" case) is an interesting topic for future research.
\\
\\
Figure \ref{fig: mc-distribution} also shows the sampling distribution of naive probit estimates of $\delta$. Such estimates, since they fail to take into account the simultaneous determination of $Y_1,\ldots,Y_T$, are inconsistent (and clearly so).
\begin{table}
       \begin{center}
           \caption{Monte Carlo Experiment Results} 
           \label{table:monte_carlo_results}
           \begin{tabular}{l r r r r r r}
                \hline \hline
			  & \multicolumn{3}{c}{\textsc{Panel A}} & \multicolumn{3}{c}{\textsc{Panel B}} \\
                & \multicolumn{3}{c}{\textsc{Many Games}} & \multicolumn{3}{c}{\textsc{Single Game}} \\
                & \multicolumn{3}{c}{$(N=100)$} & \multicolumn{3}{c}{$(N=1)$} \\
			\hline
			Number of players per game, $T$ & 20 & 20 & 20 & 500 & 500 & 500 \\
			Number of scenario draws, $S$ & 1 & 10 & 100 & 1 & 10 & 100 \\
			Number Monte Carlo replications  & 500 & 500 & 500 & 500 & 500 & 500  \\
			Mean of $\hat{\delta}$  & 0.208 & 0.200 & 0.197 & 0.206 & 0.199 & 0.197 \\
			Std. Dev. of $\hat{\delta}$ & 0.020 & 0.030 & 0.033 & 0.030 & 0.043 & 0.051\\
			Likelihood Ratio test size ($H_0 : \delta = \delta_0, \alpha = 0.05$)& 0.072 & 0.050 & 0.034 & 0.060 & 0.040 & 0.072 \\
			Confidence interval coverage ($1-\alpha = 0.95$)& 0.942 & 0.952 & 0.936 & 0.958 & 0.896 & 0.872 \\
                \hline \hline
		\end{tabular}
	\end{center}
        \underline{\textsc{Notes:}} Data generating process is as described in the main text and the Appendix \ref{app: MC_details}. Coverage is for the usual Wald-statistic-based confidence interval. Standard errors for these intervals were constructed from the Hessian matrices associated with simulated log-likelihood functions.
\end{table}
\\
\\
We close this section by observing that our Panel B Monte Carlo experiments are based upon a single game with $500$ actions. We are aware of no other maximum likelihood based estimator for games of this size.

\begin{figure}
\caption{\label{fig: mc-distribution} The MC distribution of the Scenario Estimator}
\includegraphics[scale=0.55]{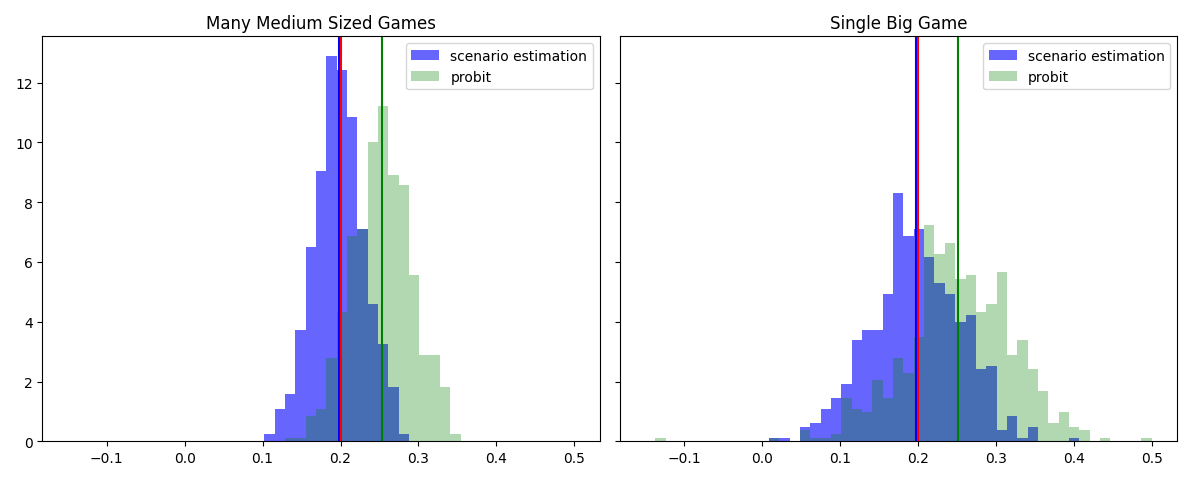}
\underline{Notes:} The figure plots a histogram estimate of the distribution of $\hat{\delta}$ across the $500$ Monte Carlo many game and single game experiments. The sampling distributions for the estimates based on $S=100$ scenario draws are the ones shown. The red vertical line marks the population value of $\delta$. The blue shaded histogram gives the distribution of our scenario-based SMLE estimates, the blue vertical line marks the average of these estimates. The green shaded histogram gives the distribution of the naive probit regression estimates of $\delta$. The probit estimates are inconsistent due to simultaneity. The green vertical line marks the average of the probit estimates. 
\end{figure}

\section{\label{sec: nykatoke} Risk-sharing in Nyakatoke} 
Joachim De Weerdt, in connection with dissertation research, collected risk-sharing links across households in Nyakatoke, a village in the Kagera Region of Tanzania (adjacent to Lake Victoria). Specifically, he asked all adult individuals in the village who they could rely upon for help and, from their responses, constructed a network of directed
links across households.\footnote{\cite{Comola_Fafchamps_EJ14} discuss the Nyakatoke dataset in detail and our interpretation of the link data follows their suggested one.} \cite{deWeerdt_IAP04} undertook a pioneering empirical analysis of these data. He modelled link formation using dyadic logistic regression methods.\footnote{See \cite{Graham_EM17} and \cite{Graham_CEMMAP22} for formal analyses of the statistical properties of these estimators in single network settings.} The analysis focused on the role of kinhsip, clan, religion, wealth and economic activity overlap in driving link formation. 
\\
\\
\cite{deWeerdt_IAP04} also posited that households exhibited a taste for transitivity in links. To capture this feature of preferences he included the number of friends in common as an additional regressor in his analyses (see Table 6 of his paper). If the payoff from $t$ directing a link to $s$ varies with the presence or absence of other links in the network (e.g., with whether $t$ and $s$ have many friends $r$ in common), then, as described earlier, the observed configuration of links will correspond to the outcome of a strategic network formation game. In such a setting, since links are simultaneously determined (and the model is also incomplete), dyadic logistic regression analysis will not deliver consistent estimates of household preferences over networks.
\\
\\
In this Section we re-visit De Weerdt's (\citeyear{deWeerdt_IAP04}) analysis. Instead of positing a taste for transitivity we, inspired by \cite{Jackson_et_al_AER12}, study whether households have a taste for ``supported links". Agent $r$ supports a link from $t$ to $s$ (and also $s$ to $t$) if both arcs $(r,t)$ and $(r,s)$ are present. \cite{Jackson_et_al_AER12} posit that agents value support, since agent $r$ can monitor and referee any relationship between $t$ and $s$. Formally we posit that household $t$'s payoff from network $\mathbf{y}$ is:
\begin{equation} \label{eq: utility_Nyakatoke}
    \upsilon\left(\mathbf{y}_{t},\mathbf{y}_{-t};x_{t},\mathbf{u}_{t},\theta\right)=\sum_{s\neq t}y_{ts}\left(x_{ts}'\beta+\delta\left[\sum_{r=1}^{T}y_{rt}y_{rs}\right]+A_{s}+B_{t}-u_{ts}\right),
\end{equation}
with $\theta=\left(\beta',\delta,\mathbf{A}',\mathbf{B}'\right)'$ for $\mathbf{A}=\left(A_{1},\ldots,A_{T}\right)'$ and $\mathbf{B}=\left(B_{1},\ldots,B_{T}\right)'$. Here $\mathbf{X}$ is a matrix of dyadic regressors. These regressors are extensively described by \cite{deWeerdt_IAP04} and more succinctly defined in Table \ref{table:nyakatoke} below. The $\left\{A_t\right\}_{t=1}^T$ and $\left\{B_s\right\}_{s=1}^T$ terms are household-specific parameters allowing for out- and in-degree heterogeneity. They are also sometimes called ego- and alter-effects or sender- and receiver-effects.
\\
\\
Such effects accommodate the reality that some households may generically get greater utility, \emph{ceterius paribus}, from directing a link, while other households may be \emph{a priori} more attractive link targets. The $\sum_{r=1}^{T}y_{rt}y_{rs}$ term counts the number of agents $r$ available to support arc $(t,s)$. The parameter $\delta$ indexes how much the payoff from directing arc $(t,s)$ increases with support. This is the ``strategic" parameter in our model and the one of primary interest here.
\\
\\
Before presenting our estimation results we note that our model is non-standard. We have complete data for $T=116$ households (out of a total of 119 households). Hence $\mathbf{Y}$ includes a total of $T(T-1)=13,340$ strategic decisions! Preferences over this graph are indexed by $\dim\left(\theta\right)=\dim\left(\beta\right)+1+2T = 243$ parameters. Inference is based entirely on a single large game (a ``fixed" $N$, ``large" $T$, ``large" $M=T-1$ analysis). The log-likelihood function does not consist of a sum of independent components, the dimension of $\theta$ grows linearly with $T$ and, finally, any large $T$ analysis would have to consider the properties of the network formation game as $T$ grows large \citep[see ][]{Menzel_ReStud2016}).
\\
\\
Inspired by \cite{Mele_Zhu_RESTAT2023} we conjecture that a triangular array set up, with $\delta$ replaced by $\delta_T=\frac{\delta}{T}$, would generate a sequence of games with a non-trivial limit as $T\rightarrow\infty$. Perhaps this set-up, paired with ideas in \cite{FernandezVal_Weidner_JOE16} and \cite{Graham_EM17}, could be used to show consistency and asymptotic normality of $\hat{\beta}$ and $\hat{\delta}$ (likely with a bias in the limit distribution). These are just conjectures; formal analysis is likely to be non-trivial and raises issues well beyond the scope of this paper. Here will simply report simulated maximum likelihood estimates (SMLEs) of $\theta$. The statistical properties of these SMLEs are, as yet, unknown. 
\\
\\
Table \ref{table:nyakatoke} reports SMLEs of $\theta$.  The simulated log-likelihood is constructed as described in Sections \ref{sec: coordination game} and \ref{sec: peer effects model} above. We use a quasi-Newton optimization algorithm, differentiating the simulated log-likelihood as detailed in Appendix \ref{app: differenication}. Standard errors, reported in parentheses, are constructed from the diagonal elements of the inverse Hessian matrix (constructed by double differentiation of the simulated log-likelihood). These standard errors have unknown statistical properties; here they simply provide measures of the curvature of our criterion function in the neighborhood of its maximum.
\\
\\
The first column of Table \ref{table:nyakatoke} reports a naive dyadic probit regression fit. These results mirror those in \cite{deWeerdt_IAP04}, who used a logit specification. The naive probit results suggest that familial connections and spatial proximity are strong drivers of link formation. There is also some evidence of religion- and wealth-based homophily. The extent of overlap in economic activities does not predict links (notwithstanding that the scope for insurance might be greater across households engaging in different types of economic activities).
\\
\\
Columns 2 to 4 of Table \ref{table:nyakatoke} report our SMLEs based on varying numbers of scenario draws ($S=1,10,100$). Our point estimates are remarkably insensitive to the number of scenarios drawn, although theoretical considerations privilege those estimates based on a larger number of scenario draws. Our discussion of the substantive aspects of the SMLEs will be based on those reported in Column 4 (which uses $S=100$ scenario draws).
\\
\\
We do not fully understand why our results are relatively stable even for small $S$, we conjecture this may reflect the following factors. First, inspection of \eqref{eq: importance_likelihood_simulator} indicates that, for a fixed $S$, our importance sampler provides an \emph{unbiased} estimate of the likelihood function for the network (the \emph{log}-likelihood function will not be unbiased). Note also that the summands in \eqref{eq: importance_likelihood_simulator} are analytically calculated $T(T-1)$-dimensional ``rectangular" probabilities (see Appendix \ref{app: differenication}). The analytical calculation of these rectangle volumes is, we believe, a distinct feature of our algorithm relative to other familiar importance samplers in econometrics (e.g., the GHK algorithm).
\\
\\
Consider the case where estimation is based on a single scenario draw (i.e., $S=1$), denote this draw by $\tilde{\mathbf{B}}^{\left(1\right)}$. This scenario is the hypercube formed by $T(T-1) = 13,340$ bucket edges, each with upper and lower bounds of, respectively, $\tilde{\bar{b}}_{ts}^{\left(1\right)}$ and $\tilde{\underline{b}}_{ts}^{\left(1\right)}$. These bucket boundaries are determined by the draws of $U_{ts}$ generated in Steps 2 and 3 of the \textsc{Scenario Sampler}. While these draws are not fully independent of one another -- the truncation point for each $U_{ts}$ draw (corresponding to a link present in the graph) depends on the values of prior draws of $U_{t's'}$ -- there is nevertheless a fair bit of independence across them.
\begin{table}
\caption{Nyakatoke SML Estimates (With Sender- and Receiver Effects)} 
\begin{center}
\begin{tabular}{lllll}
\hline  \hline
``Regressor"                         & Probit   & SMLE        & SMLE         & SMLE    \\
                                     &          & ($S=1$)     & ($S=10$)     & ($S=100$)\\
\hline
Support                              & 0.183    & 0.166       & 0.127        & 0.146   \\
($\sum_{r=1}^{T}y_{rt}y_{rs}$)       & (0.031)  & (0.015)     & (0.014)      & (0.031) \\
                                                                                       \\
Parents, children                    & 1.485    & 1.511       & 1.509        & 1.510)  \\
and siblings                         & (0.116)  & (0.113)     & (0.113)      & (0.117) \\
                                                                                       \\
Nephews, nieces, uncles, aunts,      & 0.919    & 0.897       & 0.921        & 0.929)  \\
cousins, grandparents, grandchildren & (0.128)  & (0.127)     & (0.127)      & (0.128) \\
                                                                                       \\
Other blood relative                 & 0.697    & 0.691       & 0.714        & 0.702   \\
                                     & (0.102)  & (0.100)     & (0.100)      & (0.101) \\
                                                                                       \\
Distance (km)                        & -1.375   & -1.396      & -1.420       & -1.394  \\
                                     & (0.100)  & (0.099)     & (0.099)      & (0.101) \\
                                                                                       \\ 
Same religion (Catholic,             & 0.169    & 0.156       & 0.168        & 0.172   \\
lutheran or Muslim)                  & (0.049)  & (0.048)     & (0.048)      & (0.048) \\
                                                                                       \\
Same clan                            & 0.008    & 0.018       & 0.006        & 0.011   \\
                                     & (0.079)  & (0.078)     & (0.078)      & (0.079) \\
                                                                                       \\
Both $t$ and $s$ household heads     & -0.097   & -0.082      & -0.100       & -0.097  \\
have completed primary school        & (0.156)  & (0.155)     & (0.156)      & (0.156) \\
                                                                                       \\
Activity overlap (0 to 1)            & -0.012   & -0.013      & -0.011       & -0.011  \\
                                     & (0.015)  & (0.014)     & (0.014)      & (0.015) \\
                                                                                       \\
Absolute household head age          & -0.082   & -0.080      & -0.084       & -0.081  \\
difference (decades)                 & (0.021)  & (0.021)     & (0.021)      & (0.021) \\
                                                                                       \\
Absolute wealth difference           & -0.025   & -0.024      & -0.026       & -0.025  \\
(000,000s of Tanzanian Shillings)    & (0.008)  & (0.008)     & (0.008)      & (0.008) \\
                                                                                       \\
\hline \hline
\end{tabular}
\end{center}
\underline{\textsc{Notes:}} Household-specific ego- ($A_t$) and alter- ($B_s$) effects included in all models (point estimates not reported). 
Reported standard errors constructed using the inverse of the diagonal elements of the Hessian of the simulated log-likelihood function. For
the naive probit estimates the standard errors ignore all dependence across dyads and other forms of model mis-specification.\label{table:nyakatoke}\end{table}
\\
\\
When $S=1$ the criterion function we maximize equals
\begin{equation*}
    \ln\hat{\Pr}\left(\left.\mathbf{Y}=\mathbf{y}\right|\mathbf{X};\theta\right)=\sum_{t=1}^{T}\sum_{s\neq t}\ln\left[F\left(\tilde{\bar{b}}_{ts}^{\left(1\right)}\right)-F\left(\tilde{\underline{b}}_{ts}^{\left(1\right)}\right)\right]-\ln\lambda_{\mathbf{y}}\left(\tilde{\mathbf{B}}^{\left(1\right)};\theta^{\left(0\right)}\right).
\end{equation*}
The first term is a summation of $T(T-1)$ random variables. While these random variables are not fully independent, neither are they fully dependent. Suitable normalized, it is plausible that this term has a limit as $T$ grows large. The second term in this expression, which is an output of the \textsc{Scenario Sampler}, doesn't vary with $\theta$. These considerations suggest that our simulated log-likelihood criterion may be a reasonable estimate even when $S=1$. This merits further study.
\\
\\
In looking at the point estimates, appropriately taking into account the game-theoretic details of the model results in a support coefficient about 10 percent smaller than what is produced by the naive probit fit. In contrast the coefficients on the homophily measures, on balance, increase in absolute magnitude by about 10 percent once we properly treat the network as a NE. Overall supported links do appear to generate greater utility. Links between blood relatives, geographic neighbors, co-religionists and households with similar wealth levels also generate greater utility.
\\
\\
What we wish to emphasize here is that maximum likelihood analysis, fully accounting for the complications arising from strategic interaction, is possible in a game involving over ten thousand strategic decisions and several hundred utility parameters. 

\section{\label{sec: conclusion} Conclusion}

In this paper we have presented an algorithm which facilitates simulated maximum likelihood estimation of very large binary-action supermodular games. The introduction of methods allowing for the empirical analysis of datasets recording the outcomes of strategic interaction among multiple agents is one of the major accomplishments of twenty-first century econometrics. It is our hope that the methods proposed in this paper can be used to empirically study games much larger than is currently common. 
\\
\\
Much work remains to be done. While we have shown how to compute the (simulated) MLE for some very large games, several of our examples are non-standard. These examples involve a likelihood that does not obviously factor into independent components and/or parameter spaces which grow with the ``sample size".  
\\
\\
Research on how to improve the efficiency of our importance sampler would be most welcome, as would a better understanding of how many scenario draws must be taken in practical real world settings to get reliable point estimates. It also seems likely that extant insights from the literature on simulation-based econometrics could be adapted to improve our basic approach. Some ideas in this general direction are discussed in Appendices \ref{app: differenication} and \ref{app: recycling}.
\\
\\
Our analysis involves a maintained equilibrium selection assumption. Here we have assumed that the minimal NE is the one that is played in scenarios where multiple NEs are possible. It would be easy to adapt our analysis to the case where instead the maximal equilibrium is chosen. \cite{Miyauchi_JOE16} shows that, in some models, point estimates based on these two extremal equilibria can be used to estimate the identified set for $\theta$ in the absence of any assumptions about equilibrium selection. This could be an attractive approach is settings where researchers are unwilling to maintain a particular equilibrium selection assumption.
\\
\\
In multi-action games with non-exchangeable actions the assumption that $U_{t1},\ldots,U_{tM}$ are iid is unattractive. Agents that have a taste for smoking may have, on average, a taste for drinking as well. We speculate that a pairing of our basic algorithm with ideas underlying, for example, the GHK simulator might be able to handle this extension. While this seems interesting, it is non-trivial and beyond the scope of this paper.
\\
\\
Finally, our analysis is restricted to binary supermodular games. It would also be of interest to explore whether the idea of ``scenario sampling" can be extended/adapted to non-binary choices and/or non-supermodular games (e.g., entry games with many possible entrants as in \cite{Ciliberto_Tamer_EM2009}).

\typeout{} 
\bibliography{scenario_estimation}

\newpage
\setcounter{page}{0}
\pagenumbering{arabic}
\setcounter{page}{1}
\appendix

\begin{center}
    \Large \textbf{Supplemental Web Appendix}
\end{center}

The appendix includes proofs of the theorems stated in the main text as well as some supplementary results and discussion. All notation is as established in the main text unless stated otherwise. Equation number continues in sequence with that established in the main text.

\section{\label{app: Proof-of-Theorem}Proof of Theorem \ref{thm: simulation}}

In this section we focus on the correctness of the algorithm. We elaborate in a series of lemmata the structure our algorithm relies on.  \\
\\
We order two vectors $y^1\leq y^2$ if and only if for all elements $i$ we have  $y^1_i\leq y^2_i$. With $\overline{\mathbb{R}}$ we denote the extended real line $\overline{\mathbb{R}}= \mathbb{R} \cup \{ \infty, - \infty \}$. Let $g:\{0,1\}^n \rightarrow \mathbb{R}^n$ be a supermodular, monotone increasing function and $u\in\overline{\mathbb{R}}^n$. We define the function $G_u:\{0,1\}^n \rightarrow\{0,1\}^n$ with $G_u(y)=\mathbf{1}\left(g(y)-u \geq 0 \right)$. We define the function $E:\overline{\mathbb{R}}^n \rightarrow \{0,1\}^n$ with $E(u)= (G_u)^n(0)$ where $0$  is the $n$-dimensional vector with all entries $0$ and $(G_u)^n =G_u \circ ... \circ G_u$ is the composition of $G_u$ ($n$ times). $G$ has the interpretation of a best-response function and $E$ that of an equilibrium selection (specifically it will return the minimal NE).

\begin{lemma}
$g$ is bounded.
\end{lemma}
\begin{proof}
$g$ has finite domain.
\end{proof}
\begin{lemma}
$E(u)$ is the unique minimal fixed point of $G_u$.
\end{lemma}
\begin{proof}
This follows form Tarski's fixed point theorem and the fact that the longest strictly increasing path is smaller then $n$. The uniqueness follows from supermodularity. 
\end{proof}
\begin{lemma}
$E$ is monotonically decreasing in $u$.
\end{lemma}
\begin{proof}
let $u^{*} \leq u^{**}$ then we have $G_{u^{*}}(0) \geq G_{u^{**}}(0)$. Because of the monotonicity of $G$ we have $G_{u^{*}}(G_{u^{*}}(0)) \geq G_{u^{**}}(G_{u^{*}}(0)) \geq G_{u^{**}}(G_{u^{**}}(0))$. Inductively we conclude $E( u^{*})\geq  E(u^{**})$.
\end{proof}
Next we study how the $E$ function varies as only one entry of $u$ changes. We first study $G$.
\begin{lemma}
If $u^*_j=u^{**}_j$ for $j \neq i$ then   $ G_{u^{**}}(y)_j=G_{u^{*}}(y)_j$ .
\end{lemma}
\begin{proof}
For $j\neq i$ we have $ G_{u^{**}}(y)_j= \mathbf{1} \left( g(y)_j-u^{**}_j \geq  0  \right) = \mathbf{1}\left(g(y)_j-u^{*}_j \geq 0 \right)=G_{u^{*}}(y)_j$.  
\end{proof}
\begin{corollary} \label{col:SfunctionsEquality}
If $u^*_j=u^{**}_j$ for $j \neq i$ and $G_{u^{**}}(y)_i=G_{u^{*}}(y)_i$ then $G_{u^{**}}(y)=G_{u^{*}}(y)$.
\end{corollary}
Let $e_u^i: \overline{\mathbb{R}} \rightarrow \{0,1\}^n$ with $e_u^i(x)=E(u^{'})$ where $u'_i=x$ and $u'_j=u_j$ for $j\neq i$. In words: $e_u^i(x)$ returns the minimal equilibrium given taste shock vector $u$ when $u_i$ in this vector is replaced with $x$. The next Lemma shows that for any given set of taste shocks for all actions $j \neq i$ -- $\left\{u_j\right\}_{j \neq i}$ -- we can find a threshold taste shock level $t$, such that for $u_i$ below (resp. above) this threshold it will be an equilibrium to take action $i$ (resp. not take action $i$).
\begin{lemma}\label{lemma:picewiseConstant}
There is a threshold $t$ such that $ e_u^i(x)_i= 1$ for $x \leq t$ and $ e_u^i(x)_i= 0$ for $x> t$.  Before and after the threshold $e_u^i(x)$ is constant. We denote the outcome before the threshold with  $ y^{*}=e_u^i(x)$ for $x \leq t$ and the outcome after the threshold with $ y^{**}=e_u^i(x)$ for $x> t$.  
\end{lemma}
\begin{proof}
The existence for a threshold follows directly from the monotonicity of $E$ and the shape of $G$. 
Let $x', x'' > t$ and $u', u''$ be the corresponding $u$'s when the $i^{th}$ coordinate is substituted with $x', x''$. 
We have $E(u') = (G_{u'})^n(0)= (G_{u''})^n(0)=E(u'')$. The second equality follows from corollary \ref{col:SfunctionsEquality}. This proves uniqueness of $y^{**}$.\\
The proof of uniqueness of $y^{*}$ requires a bit more work. Let $x', x'' \leq t$ with $x' <x''$ and $u' < u''$ be the corresponding $u$'s when the $i^{th}$ coordinate is substituted with $x' <x''$.  Because $E$ is monotone decreasing we have $E(u' )\geq E(u'')$. \\
Let $k$  be the first integer such that  $(G_{u'})^k(0)_i=1$ and $l$ the first integer such that  $(G_{u''})^l(0)_i=1$. By monotonicity we have $k \leq l$.  By corollary \ref{col:SfunctionsEquality}  $(G_{u'})^{k-1}(0) = (G_{u''})^{k-1}(0)$ and by monotonicity of $G$ we have $(G_{u'})^{k-1}(0) \leq  (G_{u''})^{l-1}(0)$. Now keep in mind that $u'$ and $u''$ only differ in the dimension $i$, therefore we have $(G_{u'})^{k}(0) \leq (G_{u''})^{l}(0)$. Again by monotonicity of we have $E(u'' )\geq E(u')$.
\end{proof}
Next we describe the threshold, specifically how it can be calculated. Our construction is related to the thought experiment described in the main text. First, perturb the $i^{th}$ preference shock to be so large such that it is strictly dominant not to take action $i$. Second, consider the form of the new resulting (minimal) equilibrium. The threshold is precisely the $i^{th}$ element of $g$ at this ``counterfactual" equilibrium.
\begin{lemma} \label{lemma:treshold}
The threshold of $e_u^i$ is $t= g(e^i_u(\infty))_i$.
\end{lemma}
\begin{proof}
Let $x = t-\epsilon$ and let $u' \in \overline{\mathbb{R}}^n$ with $u'_i=x$ and $u'_j=u_j$ for $j\neq i$. We have $G_{u'}(e^i_u(\infty)) = \mathbf{1}\left( g(e^i_u(\infty))-u^{'} \geq 0 \right)$. Looking at the $i$th component we have $e^i_u(\infty))_i=0$ and $G_{u'}(e^i_u(\infty))_i = \mathbf{1}\left( g(e^i_u(\infty))_i-u^{'}_i \geq 0 \right)=  \mathbf{1}\left( \epsilon \geq 0 \right)=1$. Therefore $e^i_u(x)_i = 1$ and thus $e^i_u(x)=y^{**}$. On the other hand, let $x = t+\epsilon$ and let $u' \in\overline{\mathbb{R}}^n$ with $u'_i=x$ and $u'_j=u_j$ for $j\neq i$. We have $G_{u'}(e^i_u(\infty)) = \mathbf{1}\left( g(e^i_u(\infty))-u^{'} \geq 0 \right)$. Looking at the $i$th component we have $e^i_u(\infty))_i=0$ and $G_{u'}(e^i_u(\infty))_i = \mathbf{1}\left( g(e^i_u(\infty))_i-u^{'}_i \geq 0 \right)= \mathbf{1}\left( - \epsilon \geq 0 \right)=0$. Now let $u'' \in\overline{\mathbb{R}}^n$ with $u''_i=\infty$ and $u''_j=u_j$ for $j\neq i$. Note that $ (G_{u''})^{k}(0) \leq (G_{u''})^n(0)$ for $k\leq n$.   By Corollary \ref{col:SfunctionsEquality} we have $y^{*} = e^i_u(\infty)= E(u'') = (G_{u''})^n(0)) = (G_{u'})^n(0))= E(u')=e^i_u(x)$. 
\end{proof}

With these preliminary results in hand, we can move on to show correctness of Algorithm \ref{alg: scenario-sampler}.

\begin{lemma}
Algorithm \ref{alg: threshold-finder-peer} (resp. \ref{alg: threshold-finder})  \textsc{Threshold Finder} with input $\mathbf{u}$ returns $t= g(e^i_{\mathbf{u}}(\infty))_i$.
\end{lemma}
\begin{proof}
We define $g$ as $g(\mathbf{y}) = x_{t}'\beta+\mathbf{G}_{t}\mathbf{y}\delta-u_{t}$ (resp.  $g(\mathbf{y}) = x_{tm}'\beta_{m}+s_{m}\left(\mathbf{y}_{t,-m},\mathbf{y}_{-t}\right)'\delta_{m}$). Note that $g$ is in both cases supermodular and monotone increasing. 
\end{proof}

We can now prove the first claim of Theorem \ref{thm: simulation}.

\begin{theorem} \label{thm:firstClaim}
Algorithm \ref{alg: scenario-sampler-peer} (resp. \ref{alg: scenario-sampler})  \textsc{Scenario Sampler} with input $\mathbf{y}$ returns a shock $\mathbf{u}$ such that $\mathbf{y}$ is the minimal NE at $\mathbf{u}$.
\end{theorem}
\begin{proof}
Let $n=NT$ (resp. $n=NMT$) The algorithm produces shocks $\mathbf{u}^{(0)}, \mathbf{u}^{(1)}, \mathbf{u}^{(2)},..,\mathbf{u}^{(n)}$ in step 2, shocks $\mathbf{u}^{(n+1)}, \mathbf{u}^{(n+2)},..,\mathbf{u}^{(2n)}$ in step 3 and returns $\mathbf{u}=\mathbf{u}^{(2n)}$. The claim is $E(\mathbf{u})=\mathbf{y}$. We prove this by induction. \\
\underline{Induction start}: $\mathbf{u}^{(0)}$ has only $\infty$ or $-\infty$ as entries. For any $\mathbf{y}'\in\{0,1\}^n$ we have $G_{\mathbf{u}^{(0)}}(\mathbf{y}') =\mathbf{1}\left(g(\mathbf{y}')-\mathbf{u}^{(0)} \geq 0 \right) = \mathbf{y}$ and therefore $(G_{\mathbf{u}^{(0)}})^n(0)= \mathbf{y} = E(\mathbf{u}^{(0)})$. \\
\underline{Induction step}: The claim is true for $\mathbf{u}^{(0)}, \mathbf{u}^{(1)}, \mathbf{u}^{(2)},..,\mathbf{u}^{(k)}$. If $y_{k+1}=1$ we have $u^{(k)}_{k+1} = -\infty$ and $u^{(k+1)}_{k+1} \in (-\infty,t]$. By Lemma \ref{lemma:picewiseConstant} and Lemma \ref{lemma:treshold} we have $E(\mathbf{u}^{(k)}) = E(\mathbf{u}^{(k+1)})$. Similarly, if $y_{k+1}=0$ we have $u^{(k)}_{k+1} = \infty$ and $u^{(k+1)}_{k+1} \in (t,\infty)$. By Lemma \ref{lemma:picewiseConstant} and Lemma \ref{lemma:treshold} we have $E(\mathbf{u}^{(k)}) = E(\mathbf{u}^{(k+1)})$.
\end{proof}
It remains to be shown that
\begin{theorem}\label{thm:secondClaim}
For every shock $\mathbf{u}$ with minimal equilibrium $\mathbf{y}$, it is possible that Algorithm \ref{alg: scenario-sampler-peer} (resp. \ref{alg: scenario-sampler}) \textsc{Scenario Sampler} with input $\mathbf{y}$ returns shock $\mathbf{u}$.
\end{theorem}
\begin{proof}
Let $\mathbf{u}$ be a taste shock with minimal equilibrium $\mathbf{y}$. We first look at Step 2. Note that for $y_k=0$ we have $0=e^k_{\mathbf{u}}(\infty)_k = e^k_{\mathbf{u}}(u_k)_k = 0$ and $e^k_{\mathbf{u}}(\infty) = e^k_{\mathbf{u}}(u_k) = \mathbf{y}$. Now $u_k> g(\mathbf{y})_k$ the threshold is $t= g(e^k_{\mathbf{u}}(\infty))_k = g(\mathbf{y})_k$ so $u_k>t$.
\\
Now we look at Step 3: we have $\mathbf{u}^{(n+1)}, \mathbf{u}^{(n+2)},..,\mathbf{u}^{(2n)} \leq \mathbf{u}$. Furthermore we know that for all $y_k=1$ the equilibrium/fixed-point condition implies that shock $u_k$ fulfills $u_k\leq g(e^k_{\mathbf{u}}(\infty))_k = t$. Now $g$ is monotone increasing and $E$ is monotone decreasing, therefore the function $T:\overline{\mathbb{R}}^n \rightarrow \mathbb{R}$ with $T(\mathbf{u}) = g(e^k_{\mathbf{u}}(\infty))_k$ is monotone decreasing. Therefore $u_k\leq g(e^k_{\mathbf{u}}(\infty))_k= T(\mathbf{u}) \leq T(\mathbf{u}^{(k)})= g(e^k_{\mathbf{u}^{(k)}}(\infty))_k$. Note that $T(\mathbf{u}^{(k)})$ is the threshold for the algorithm in iteration $k$ of Step 2.\\
We conclude that it is possible that  Algorithm 1 (resp. \ref{alg: scenario-sampler}) \textsc{Scenario Sampler} to draw $\mathbf{u}$.
\end{proof}
The proof of Theorem \ref{thm: simulation-peer} (resp. \ref{thm: simulation-peer}) in the main text follows directly from Theorem \ref{thm:firstClaim}, which proves the first claim, and Theorem \ref{thm:secondClaim}, which proves the second claim.

\section{\label{app: MC_details} Details of Monte Carlo experiments}

The Monte Carlo design uses a random geometric network to construct $\mathbf{D}$. The friendship network is exogenous and determines who is a peer of whom. Specifically agents are scattered uniformly on the two-dimensional plane$\left[0,\sqrt{T}\right]\times\left[0,\sqrt{T}\right]$.The network is then generated according to the rule $D_{st}=\mathbf{1}\left(A_{st}-U_{st}\geq0\right)$, with $U_{st}$ logistic and $A_{st}$ taking one of two values. If the Euclidean distance between agents $s$ and $t$ is less than or equal to $r$, then $A_{st}=\ln\left(\frac{0.75}{1-0.75}\right)$, otherwise $A_{st}$ equals negative infinity. This calibration means that agents link to anyone less than $r$ apart with probability $0.75$, while those greater than $r$ apart link with probability zero. Self-links are not allowed such that the diagonal elements of $\mathbf{D}$ are all equal to zero. This basic design features in \cite{Graham_NBER16}.
\\
\\
The expected out- and in-degree of a randomly sampled agent in the resulting network is approximately $0.75\pi r^{2}$. We set $r=10/0.75\pi$ such that average degree is approximately $10$. With these parameter values almost all agents are part of one giant component.
\\
\\
The best reply function for agents $t=1,\ldots,T$ is

\begin{equation} \label{eq:mcSistemEq}
    Y_t =  1( \beta_1 X_{1t}+\beta_2 X_{2t} + \beta_3 X_{3t} + \beta_4 X_{4t}  + \delta \mathbf{D}_t'\mathbf{Y}   - U_t)  > 0),
\end{equation}
\\
with $U_t$ a normally distributed random utility shock. We draw $X_{1t}, X_{2t}$  with probability $\frac{1}{2}$ from $\{0,1\}$, as group indicator variables. We draw  $X_{3t}, X_{4t}$  uniformly from $[0,1]$ as continuous covariates. The configuration of covariates across agents, that is the $\mathbf{X}$ matrix, is held fixed across Monte Carlo replications. We set the parameters $\beta_1 =-1, \beta_2=-0.5, \beta_3=-1, \beta_4 = 0.5$ and $\delta = 0.2$. The above system of equations has multiple solutions; we assume that agents play the minimal equilibrium (i.e., the solution where the fewest number of agents take the action).
\\
\\
For each simulation, we draw the utility shocks $\mathbf{U}$ and then find the minimal solution, $\mathbf{Y}$, to \eqref{eq:mcSistemEq} using fixed point iteration \citep{Tarski_PJM55}. We estimate $\beta$ and $\delta$ by simulated maximum likelihood (SML) using the observed values of $\mathbf{X}$ and $\mathbf{Y}$. 
\\
\section{\label{app: differenication}Differentiability} 

In this section we show how the simulated log likelihood can be differentiated. This facilitates the use of gradient-based optimization methods and, consequently, allows a researcher to fit models where the dimension of $\theta$ is non-trivial. It also allows for standard error computation via the log-likelihood Hessian matrix. As noted by \cite{McFadden_EM89}, \cite{Ackerberg_QME2009} and others, one advantage of importance sampling in the SML context is that the simulated log-likelihood function is differentiable. In this sense, the arguments in this Appendix are not especially novel. However we include them as they are useful for implementation and our discussion provides additional insight into our overall approach to estimating ``large" games.
\\
\\ 
First, quickly, some high level intuition for what we are doing. Our estimator considers a number of sampled scenarios. As Figure \ref{fig: understanding-scenarios} shows, a scenario can be thought of as a rectangle shaped area (i.e., a hypercube). The borders of these rectangles vary with $\theta$, either increasing or decreasing the \emph{ex ante} probability attached to any given scenario. The velocity of these shifts can be traced out, allowing for the calculation of the derivative.
\\
\\ 
We start with a reformulation of equation \eqref{eq: importance_likelihood_simulator} of the main text:

\begin{equation*}
\hat{\Pr}\left(\left.\mathbf{Y}=\mathbf{y}\right|\mathbf{X};\theta\right)=\frac{1}{S}\sum_{s=1}^{S}\frac{\int_{\mathbf{u}\in\tilde{\mathbf{B}}^{\left(s\right)}} f_{\mathbf{U}}\left(\mathbf{u}\right)\mathrm{d}\mathbf{u}}{\lambda_{\mathbf{y}}\left(\tilde{\mathbf{B}}^{\left(s\right)};\theta^{\left(0\right)}\right)}.
\end{equation*}
Differentiating with respect to $\theta$ yields:
\begin{equation*}
\frac{\hat{\partial\Pr}\left(\left.\mathbf{Y}=\mathbf{y}\right|\mathbf{X};\theta\right)}{\partial\theta}=\frac{1}{S}\sum_{s=1}^{S}\frac{\frac{\partial}{\partial\theta}\int_{\mathbf{u}\in\tilde{\mathbf{B}}^{\left(s\right)}}f_{\mathbf{U}}\left(\mathbf{u}\right)\mathrm{d}\mathbf{u}}{\lambda_{\mathbf{y}}\left(\tilde{\mathbf{B}}^{\left(s\right)};\theta^{\left(0\right)}\right)}.
\end{equation*}
Note the value of $\lambda_{\mathbf{y}}\left(\tilde{\mathbf{B}}^{\left(s\right)};\theta^{\left(0\right)}\right)$ is determined \emph{a priori} and does not vary with $\theta$. Only the numerator of \eqref{eq: importance_likelihood_simulator} varies with $\theta$. Note it very well may make sense to adjust the algorithm such that it works better for certain values of $\theta$. One heuristic is to begin with some $\theta^{\left(0\right)}$ in the denominator of \eqref{eq: importance_likelihood_simulator} and construct a preliminary estimate of $\theta$ using just a few scenario draws, $S$. Call this estimate $\tilde{\theta}$. Then, in a second round of estimation, set $\theta^{\left(0\right)}=\tilde{\theta}$, and compute a new, more accurate, estimate of $\theta$ using a larger number of scenario draws, $S$.
\\
\\
We require the derivative of the integral
\begin{equation*}
    \int_{\mathbf{u}\in\tilde{\mathbf{B}}^{\left(s\right)}}f_{\mathbf{U}}\left(\mathbf{u}\right)\mathrm{d}\mathbf{u}.
\end{equation*}
Geometrically this integral is over a multi-dimensional rectangle. Indeed, as shown earlier, it has a ``closed form" expression of
\begin{equation*}
     \Pr\left(\left.\tilde{\mathbf{B}}^{\left(s\right)}=\mathbf{b}\right|\mathbf{X};\theta\right)	=  \int_{\mathbf{u}\in\mathbf{b}}f_{\mathbf{U}}\left(\mathbf{u}\right)\mathrm{d}\mathbf{u}
	=  \prod_{t=1}^{T}\prod_{m=1}^{M}\left[F\left(\bar{b}_{tm}\right)-F\left(\underline{b}_{tm}\right)\right],
\end{equation*}
where, recalling the notation from Section \ref{sec: coordination game}, $\bar{b}_{tm}$ and $\underline{b}_{tm}$ are the upper and lower bucket boundaries for agent $t$'s $m^{th}$ action in scenario $\mathbf{b}$. By the product rule of differentiation we get:
\begin{equation*}
\frac{\partial\Pr\left(\left.\tilde{\mathbf{B}}^{\left(s\right)}=\mathbf{b}\right|\mathbf{X};\theta\right)}{\partial\theta}=\Pr\left(\left.\tilde{\mathbf{B}}^{\left(s\right)}=\mathbf{b}\right|\mathbf{X};\theta\right)\left[\sum_{t=1}^{T}\sum_{m=1}^{M}\frac{f\left(\bar{b}_{tm}\right)\frac{\partial\bar{b}_{tm}}{\partial\theta}-f\left(\underline{b}_{tm}\right)\frac{\partial\underline{b}_{tm}}{\partial\theta}}{F\left(\bar{b}_{tm}\right)-F\left(\underline{b}_{tm}\right)}\right],
\end{equation*}
where $\bar{b}_{tm}$ and $\underline{b}_{tm}$ are affine functions of $\theta$ with simple derivatives (see Section \ref{sec: coordination game} of the main text). Returning to the high level intuition introduced above, $\frac{\partial\bar{b}_{tm}}{\partial\theta}$ is the velocity with which the scenario boundaries are moving, while $f\left(\bar{b}_{tm}\right)$ is the density at the boundary, itself a measure on how much probability is lost or gained as the border changes.  The product of these two terms equals the change in probability (re-scaled according to the product rule). 
\\
\\
Putting it all together we get the total derivative of the log likelihood. While there are many summands involved, their total is polynomial in the number of parameters, the number of scenarios, the number of players and the number of actions of each player. This make the method computationally feasible. 
\\
\\
As a final note we consider the case when the researcher observes $N$ independent games. In this case the it is possible to draw scenarios separately for each game. This facilitates the computation as well as the derivation of asymptotic properties. With $N$ games the simulated log-likelihood takes the form
\begin{equation}
\ln \hat{\Pr}\left(\left.\mathbf{Y}= \mathbf{y}\right|\mathbf{X};\theta\right)=\sum_{i=1}^{N}\ln\hat{\Pr}\left(\left.\mathbf{Y}_{i}= \mathbf{y}_i\right|\mathbf{X}_{i};\theta\right)\label{eq: logL}
\end{equation}
with gradient vector
\begin{equation}
\frac{\partial \ln \hat{\Pr}\left(\left.\mathbf{Y}= \mathbf{y}\right|\mathbf{X};\theta\right)}{\partial\theta}=\sum_{i=1}^{N}\frac{1}{\hat{\Pr}\left(\left.\mathbf{Y}_{i}= \mathbf{y}_i\right|\mathbf{X}_{i};\theta\right)}\frac{\partial\hat{\Pr}\left(\left.\mathbf{Y}_i=\mathbf{y}_i\right|\mathbf{X};\theta\right)}{\partial\theta}.\label{eq: score}
\end{equation}

\section{\label{app: recycling} Recycling}

Efficient implementation of a simulation-based estimator involves a number of subtle details. Equation \eqref{eq: importance_likelihood_simulator} is an average over scenarios. Whether an underlying sets of random utility shocks lies in one scenario or another varies with the scenario boundaries, themselves functions of the underlying parameter, $\theta$. In this Appendix we discuss some practical implications of this observation for estimation.
\\
\\
As in other applications of simulated maximum likelihood, to avoid criterion function ``chatter" the researcher should begin estimation by drawing $S \times T \times M$ independent standard uniform random variables. These draws are subsequently held fixed as estimation proceeds. To recompute the scenarios $\tilde{\mathbf{B}}^{(s)}$ for $s=1,\dots,S$ as $\theta$ changes during optimization, the researcher uses these initial simulated uniform random variables in conjunction with the inverse probability transform method in Steps 2.a and 3.a.ii of Algorithm \ref{alg: scenario-sampler} to generate new preference shocks. This method keeps the underlying simulation randomness constant as optimization proceeds (making the target function smooth in $\theta$), in turn allowing for the discovery of a well-defined criterion function maximum. This approach to SMLE is standard.
\\
\\
In this section we discuss how, in some settings, it is possible to speed up computation further by ``recycling" scenarios. This can substantially reduce the computational cost of the optimization.
\\
\\
To understand how it is possible to reuse, or recycle, the scenarios is is essential to have a closer look at how they are constructed. The scenarios are constructed in step 3.b of the \textsc{Scenario sampler} algorithm (Algorithm \ref{alg: scenario-sampler}). One way to rigorously find the scenario, say, $\tilde{\mathbf{B}}$ associated with a given vector of taste shocks, $\mathbf{U}$ is as follows. First, evaluate  $g_{tm}(\mathbf{y},x_{tm})$ for all $\mathbf{y}$, to find the bucket borders for action $m$ of agent $t$ (do this is for all $TM$ strategic actions). Second, for each coordinate of $\mathbf{U}$ find the lower and upper boundaries of the bucket into which it falls. The buckets for each $U_{tm}$ are then combined to find the scenario associated with the full utility shock vector, $\mathbf{U}$. This method is not feasible, however, because $\mathbf{y}$ can take exponentially many values. Fortunately, in practice, it is usually possible to apply some heuristics and find the right bucket for each $U_{tm}$ in reasonable time. For example, in the peer effects model a bucket corresponds to a utility shock range where \textit{if $k$ friends of agent $t$ take the action, agent $t$ also takes the action}. In this case, for a given realization $U_{tm}$, we can simply look up how many friends are needed for agent $t$ to take the action and get the required bucket boundaries.
\\
\\
The mapping from preference shocks to scenarios gets more complicated in models with multiple strategic parameters. For example, consider the network formation game where agents value both reciprocity and transitivity in links. Here you could find the bucket in which $U_{tm}$ lies by (i) considering the number of transitive triads $t$ needs in order to form the link with $s$ when the $(t,s)$ arc is reciprocated and (ii) repeating the same computation for the case where the $(t,s)$ arc is unreciprocated. In both cases we get a partition of the support of $U_{tm}$. We then can find the buckets by taking the intersections of all the intervals in the two partitions. From this discussion, it is apparent that the computational cost increases strongly with the number of strategic parameters. However, in most application there will only be a few strategic parameters and computation remains feasible. Indeed, in all of our examples it is possible to construct a heuristic for finding the scenario associated with a shock realization efficiently. The major part of the computational cost associated with estimation involves the repeated updating/simulation of the utility shocks (and hence the scenarios) as $\theta$ changes. 
\\
\\
Lets consider the case with only one strategic parameter in detail. Assume $g_{tm}(\mathbf{y},x_{tm}) =x_{tm}'\beta_{m}+s\left(\mathbf{y}_{t,-m},\mathbf{y}_{-t}\right)\delta$ with $\delta>0$. The set $\mathcal{L}_{tm}=\{s\left(\mathbf{y}_{t,-m},\mathbf{y}_{-t}\right) | \mathbf{y}_{t,-m}\in\{0,1\}^{M-1}\ \&\ \mathbf{y}_{-t}\in\{0,1\}^{(T-1)M}\}$ is finite. We can enumerate it starting from the smallest value $l_{mt}^1$ to the biggest value; the elements of $\mathcal{L}_{tm}$ partition the real line into $\#\mathcal{L}_{tm}+1$=$L_{tm}+1$ intervals. The same logic applies to all values of $t$ and $m$. By taking the Cartesian product of these partitions we get a partition of $\mathbb{R}^{TM}$, the support of the random utility preference shocks, $\mathbf{U}$. This partition can be enumerated explicitly by using the enumeration of the intervals of the real lines lexicographically.
\\
\\
Observe that the function $g$ scales the values of $\mathcal{L}_{tm}$ by a positive value and then translates the values. Both of these operations are order preserving: there is a one-to-one correspondence of the partition of $\mathbb{R}^{TM}$ constructed above and the scenarios $\mathbb{B}$. This is very useful since it implies that the number of scenarios does not depend on the value of $\theta$. Furthermore, the scenario boundaries are affine functions of $\theta$. In particular the bucket boundaries are are continuous and differentiable in $\theta$.
\\
\\
Recall equation \eqref{eq: importance_likelihood_simulator} of the main text
\begin{equation*}
\hat{\Pr}\left(\left.\mathbf{Y}=\mathbf{y}\right|\mathbf{X};\theta\right)=\frac{1}{S}\sum_{s=1}^{S}\frac{\zeta\left(\tilde{\mathbf{B}}^{\left(s\right)};\theta\right)}{\lambda_{\mathbf{y}}\left(\tilde{\mathbf{B}}^{\left(s\right)};\theta^{\left(0\right)}\right)}.
\end{equation*}
The only part of this expression depending on the parameter $\theta$ is $\zeta\left(\tilde{\mathbf{B}}^{\left(s\right)};\theta\right)$. With only one strategic parameter $\zeta\left(\tilde{\mathbf{B}}^{\left(s\right)};\theta\right)$ is a smooth function of $\theta$. This makes it possible to evaluate the estimated likelihood at different parameter values without the need to repeatedly sample new scenarios.  By reusing (or recycling) the scenarios during optimization, the computation of an iteration is reduced substantially. 
\\
\\
For the case when there is more than one strategic interaction parameter this does not work. The easiest way to see this is to consider again the network formation game where the scenarios are constructed by using the intersection of intervals. The number of scenarios in this case will depend on the values of $\theta$. As a result the estimated likelihood is only locally differentiable and it is not possible to reuse the scenarios at a different value of $\theta$. 
\\
\\ 
Scenario recycling is related to the change-of-variables method proposed in \cite{Ackerberg_QME2009} to minimize function evaluation in complex structural econometric models. His method also appears to be limited to models with a single strategic parameter (see part 1 of his \textrm{CS} assumption). Although related, ``scenario recycling", as described here does not appear be a special case of his approach.
\end{document}